\documentclass[aps,pra,twocolumn,floatfix,reprint]{revtex4-2}
\usepackage[page]{appendix}
\usepackage{float}
\usepackage{hyperref}
\hypersetup{
    colorlinks=true,
    linkcolor=blue,
    filecolor=magenta,
    urlcolor=blue,
}
\usepackage{bbm,enumitem}
\usepackage{braket}
\usepackage{amsmath, amsfonts, amsthm, amssymb, bbold, pifont, cleveref}
\usepackage{url}
\usepackage{graphics,color,psfrag}
\usepackage{subcaption}
\usepackage{breqn}
\usepackage{makecell}
\usepackage[table]{xcolor}
\usepackage{multirow}
\usepackage{algorithm}
\usepackage[noend]{algpseudocode}

\newtheorem{thm}{Theorem}
 \numberwithin{dummy}{thm}
\newtheorem{fact}{Fact}
\newtheorem{lem}{Lemma}

\newtheorem{cor}{Corollary}[thm]

\newtheorem{remark}{Remark}

\newcommand{\Tr}[1]{\operatorname{Tr} \left\{ #1 \right\}}

\DeclareMathAlphabet{\gcal}{OMS}{cmsy}{m}{n}

\newcommand{\EE}{\mathbb{E}}

\DeclareMathOperator*{\argmax}{argmax}

\newcommand{\K}{\mathbf{K}}
\newcommand{\bK}{\mathbf{K}}
\newcommand{\R}{\mathcal{R}}
\newcommand{\OO}{\mathcal{O}}

\newcommand{\myS}{{\cal S}}
\newcommand\norm[1]{\left\lVert#1\right\rVert}
\DeclareMathOperator{\Var}{\textrm{Var}}

\DeclareMathOperator{\T}{\intercal}

\def\bx{{\mathbf x}}

\newtheorem{definition}{Definition}

\def\subg{\mathbf{\tt SubG}}
\def\P{\mathcal P}
\def\N{\mathbb N}

\def\D{\mathcal D}

\def\X{{\mathcal X}}

\def\bx{{\mathbf x}}

\def\bv{{\mathbf v}}
\def\psdset{{\mathbb S^{m}_{+}}}
\def\RR{{\mathbb R}}
\def\c{{\mathrm c}}
\def\Prob{\mathrm {Prob}}
\def\L1{\mathrm {L1}}
\def\prob{\mathrm {Prob}}
\def\sign{\mathrm {sgn}}
\def\Acc{\mathrm {Acc}}
\def\test{\mathrm {test}}
\def\rob{\mathrm {rob}}

\def\train{\mathrm {train}}

\def\emp{\mathrm {emp}}
\def\trials{\mathrm {trials}}
\def\target{\mathrm {target}}

\def\sv{\mathrm {sv}}
\def\IE{\textrm{IE}}
\def\bin{\textrm{Bin}}
\def\theory{\mathrm {sg}}
\def\practical{\mathrm {practical}}

\begin{document}
\title{Shot-frugal and Robust quantum kernel classifiers}

\author{Abhay Shastry }
\altaffiliation{Department of Computer Science and Automation, Indian Institute of Science, Bangalore 560012}
\email{abhayshastry@arizona.edu}
\author{Abhijith Jayakumar}
\altaffiliation{Theoretical Division, Los Alamos National Lab, New Mexico, US 87545}
\author{Apoorva D. Patel}
\altaffiliation{Center for High Energy Physics, Indian Institute of Science, Bangalore 560012}
\author{Chiranjib Bhattacharyya}
\altaffiliation{Department of Computer Science and Automation, Indian Institute of Science, Bangalore 560012}

\begin{abstract}
Quantum kernel methods are a candidate for quantum speed-ups in supervised machine learning.
The number of quantum measurements $N$ required for a reasonable kernel estimate is a critical resource, both from complexity considerations and because of the constraints of near-term quantum hardware.
We emphasize that for classification tasks, the aim is {\em reliable} classification and {\em not precise} kernel evaluation, and demonstrate that the former is far more resource efficient.
Furthermore, it is shown that the accuracy of classification is not a suitable performance metric in the presence of noise and we motivate a new metric that characterizes the reliability of classification. We then obtain a bound for $N$ which ensures, with high probability, that classification errors over a dataset are bounded by the margin errors of an idealized quantum kernel classifier.
Using chance constraint programming and the subgaussian bounds of quantum kernel distributions, we derive several Shot-frugal and Robust (ShofaR) programs starting from the primal formulation of the Support Vector Machine. This significantly reduces the number of quantum measurements needed and is robust to noise by construction.
Our strategy is applicable to uncertainty in quantum kernels arising from {\em any} source of unbiased noise.
\end{abstract}

\maketitle

\section{Introduction}

Kernel methods have been used extensively in a variety of machine learning tasks like classification, regression, and dimensionality reduction \cite{ScholkopfBook}.
The core ingredient of such algorithms is the so-called kernel trick, which implicitly embeds a datapoint into a higher (may be infinite) dimensional inner product space, known as the \emph{feature space}.
The embedding is implicit because the vectors in the feature space are never explicitly computed. Instead, a \emph{kernel function} is used to map any two input data points to the  inner product between their embeddings in the feature space.
The key insight that allows kernel methods to be used with quantum computers is the identification of a vector in the feature space with the density matrix of a many-body quantum system \cite{Schuld2019,Schuld2021}.

Quantum computing allows us to efficiently compute the inner products of exponentially large density matrices without the explicit tomography of the density matrices themselves.
This has prompted researchers to search for a quantum advantage in kernel methods  \cite{Kubler2021inductive,Havlicek2019,Liu2021rigorous,Huang2021data,Wang2021nisq}. However, quantum kernels introduce new challenges that are not present in the classical setting. In a major departure from classical kernels, quantum kernels cannot be explicitly evaluated. They can only be estimated by repeatedly preparing and measuring quantum logic circuits. Quantum measurements introduce a fundamental stochasticity (sampling noise) into the problem, even for an ideal device.
This is a crucial difference as the number of circuit evaluations one can perform is a finite resource in all available quantum computing platforms. In particular, the number of circuit evaluations $N$ required for an accurate kernel estimate up to a precision $\varepsilon$ scales as $\varepsilon^{-2}$. Exponential precision (in terms of the number of bits of $\varepsilon$) in the kernel estimate thus imposes exponentially large requirements on $N$, thereby making it hard to claim a quantum advantage \cite{Kubler2021inductive,thanasilp2022exponential}. A focus on this requirement is especially relevant for the current noisy intermediate-scale quantum (NISQ) era \cite{Bharti2022noisy} quantum devices, where the number of measurement shots are a resource to be used frugally \cite{Kubler2020frugal,Arrasmith2020frugal,Gu2021}.

In this work we address binary classification using a \textit{Support Vector Machine} (SVM) along with quantum kernels, emphasizing the number of circuit evaluations ($N$) required for classification. After choosing a suitable feature map, an SVM training algorithm builds a separating hyperplane between the classes of datapoints by solving a convex optimization problem.
In an SVM, the boundary between the classes carries the information relevant to the classifier and is described by a set of support vectors (of size $m_{\sv}$).
Together with the kernel trick, SVMs efficiently construct highly nonlinear classifiers in the original data space $\X$ \cite{Cortes1995support,Cervantes2020comprehensive,Burges1998tutorial}.
Suppose a kernel function $K:\X\times\X\to\RR$, which may be classically hard to compute, has to be evaluated on a quantum machine. How many measurements would be needed to construct a reliable classifier?
Crucially, our goal is {\em not} to accurately evaluate the kernel function itself
but to construct a reliable classifier with it.

\subsection{Problem Setup}

We work in the hybrid quantum-classical setup \cite{Havlicek2019,Schuld2019,Blank2020}.
\begin{enumerate}
\item The training phase requires kernel function evaluations between each pair of input datapoints (training kernel matrix).
This matrix specifies the SVM optimization problem which is solved on a classical machine. Its solution gives the optimal parameters $\beta,b$ of the SVM classifier.

\item A datapoint $\bx\in\X$ is then classified by the classifier function $f(\bx;\beta,b)$ which requires the estimation of kernel functions for each pair $K(\bx,\bx_{i}),\ i\in[m_{\sv}],$ consisting of the given datapoint and the set of support vectors of the SVM. Kernel evaluations are done on a quantum machine by repeated measurements of a quantum logic circuit. This measurement process introduces stochasticity (sampling noise) to the kernel estimates.

\item Using $N$ measurements {\em per kernel estimate}, we denote the stochastic kernel classifier (\ref{eq:skc*}, Definition \ref{def:skc}) by $f^{(N)}(\bx;\beta,b)$. Our aim is to study how reliably it reproduces the same label as the ideal or exact kernel classifier (\ref{eq:ekc*}, Definition \ref{def:ekc}) $f^{*}(\bx;\beta,b)$, which is notionally obtained when $N\to\infty$.
\end{enumerate}
Our study concerns the effect of noise during classification. Noise may or may not be present during training.

\subsection{Contributions}
In Section \ref{sec:q_samp_correction}, we highlight the problem of unreliability of quantum classifiers and show that classification accuracy (Definition \ref{def:accuracy}) is not a suitable error metric in the presence of noise (see Fig.\ \ref{fig:prob_statement}). We motivate
the metric of {\em reliability} (Definition \ref{def:reliability}) which is the {\em probability} with which the \ref{eq:skc*} agrees with \ref{eq:ekc*}.

In Section \ref{sec:q_circuits}, two quantum circuits are given to estimate the kernel, namely, the SWAP test and the GATES test. We show in Appendix \ref{sec:subgaussian} that the SWAP test always leads to greater variance for the same $N$. In Definition \ref{lemma:cf}, we quantify this by introducing the circuit factor.

In Theorem \ref{thm:N_bound}, we derive a lower bound on $N$ above which the SKC reliably classifies a given datapoint (i.e.\, with high probability agrees with EKC). This result sets up the margin of classification $\gamma$ as the relevant parameter to determine the number of measurement shots required as opposed to the precision $\varepsilon$ of a kernel estimate.

In this context, our main contributions, given in Sections \ref{sec:expt_1} and \ref{sec:expt_2}, are summarized below.\\
{\em Assumption 1a}: The training kernel matrix is assumed to be exact (noiseless training phase; see Appendix \ref{sec:more_rel} for an empirical justification).

\begin{enumerate}

\item Theorem \ref{thm:emp_risk} characterizes the number of measurement shots needed to reliably reproduce the labels of an ideal classifier (EKC) over any dataset. Based on concentration bounds of subgaussian distributions, Theorem \ref{thm:emp_risk} gives a lower bound $ N_{\theory}$, which scales as $ m_{\sv}\log M/\gamma^{2}$ for a dataset of size $M$. For any $N\geq N_{\theory}$,
the SKC has a classification error which is bounded by the $\gamma$-margin error of the EKC.
The margin $\gamma$, unlike the precision $\varepsilon$, is {\em not} vanishingly small for a well-chosen kernel.
Precise kernel evaluations, on the other hand, lead to a shots requirement which scales as $m_{\sv}M/\varepsilon^{2}$ (see Appendix \ref{sec:precise} for a derivation).

This result paves the way to adaptively set the shots requirement based on the dataset and kernel choice.

\item In Section \ref{sec:expt_2}, we take the chance constraint approach to handle the uncertainty in the training kernel matrix starting from the {\em primal formulation} of the SVM. In Theorem \ref{thm:rob}, we derive a convex approximation to the chance constraint and introduce the Shot-frugal and Robust (\ref{eq:rob_primal}) program. A robust stochastic kernel classifier (RSKC) is constructed using this program which reliably reproduces EKC with far fewer measurement shots than would be needed for the nominal SKC. This program results in a classifier which is robust to noise by construction.
\end{enumerate}
{\em Assumption 1b:} We now relax the Assumption 1a for the Robust program but the nominal SVM program still has access to the exact (noiseless) training kernel matrix.
\begin{enumerate}[resume]
\item In Section \ref{sec:stat_est}, we consider the setting where the training kernel matrix is also inferred from repeated measurements.
In Theorem \ref{thm:rob_est}, we derive a refined Shot-frugal and Robust \ref{eq:J_rob_l2_est} program that takes as input an estimated training kernel matrix in a principled manner.
In Corrolaries \ref{thm:rob_l1} and \ref{thm:rob_est_l1}, we consider the L1-norm relaxation of the two ShoFaR programs,   \ref{eq:J_rob_l1} and \ref{eq:J_rob_l1_est} respectively, in order to reduce the set of support vectors. We compare the performance of all these programs to the nominal \ref{opt:p} SVM and show significant savings (see Table\ \ref{tab:compare_robust}), {\em even though} the nominal SVM was given access to the exact training kernel matrix. 
\end{enumerate}
{\em Assumption 1c:} We now relax the Assumption 1b for both the nominal SVM and the Robust programs and additionally include a device noise for both training and classification phases.
\begin{enumerate}[resume]
\item In Appendix \ref{sec:depol}, we work under the physically realisitic assumption above. The device noise is modeled using the depolarizing noise model and the noise affects both the classification and training phases. Without error mitigation during training, we note that the accuracy of EKC suffers for the nominal SVM. However, the accuracy of the Robust SVM is not affected with the same noise. It is noteworthy that the Robust program \ref{eq:J_rob_l2_est}, without any modifications or error mitigation, outperforms the nominal SVM (both on accuracy and reliability) even if the latter has undergone error mitigation.
\end{enumerate}

Together, these contributions highlight the problem of unreliability of quantum kernel classifiers and develop the methods to construct robust classifiers which use quantum resources frugally \footnote{The robust classifiers are constructed by solving a second-order cone program rather than a quadratic program and may thus require more classical resources.}.

\subsection{Datasets}
We present our results over the following datasets which have been frequently used in the quantum machine learning literature:
\begin{enumerate}
\item Circles ({\tt make_circles} in {\tt sklearn}) \cite{Schuld2019, vedaie2020quantum, Watkins2023, Yang2023}
\item Moons ({\tt make_moons} in {\tt sklearn}) \cite{Lloyd2020,Simonetti2022,Watkins2023,Banchi2021}
\item Havlicek dataset \cite{Havlicek2019}, generated over 2 qubits
\item Checkerboard dataset \cite{Hubregtsen2021}
\end{enumerate}
All the figures in the article correspond to the Circles dataset.
Our results in the main Sections \ref{sec:expt_1}, \ref{sec:chance} and Appendix \ref{sec:depol} are presented over the training dataset itself.
This choice starkly highlights the problem of unreliable classification. It shows that
SKC is incapable of reliably reproducing {\em even the training labels} when $N$ is not large enough.
In Section \ref{sec:stat_est} and Appendices \ref{sec:datasets} and \ref{sec:more_rel}, we show the performance of our Robust programs over an independent test set.

\section{Preliminaries}
\label{sec:prelims}
In this Section, we introduce the relevant notation and provide useful preliminaries.

{\em Notation:} The set of natural numbers is denoted by $\N$.
For any $m \in \N$, $[m]$ denotes the set $\{1,2,\ldots,m\}$.
$\RR^d$ denotes a $d$-dimensional real vector space.
For any $\bv \in \RR^d$, the euclidean norm is $\|\bv\| = \sqrt{\bv^\top \bv}$.
A random variable $X \in \RR^d$ distributed according to probability distribution ~$\P$ is denoted as $X \sim \P$, and its expectation value denoted as $\EE(X)$.
The binomial random variable is denoted by $\bin(N,p)$, where $N \in \N$ is the number of independent Bernoulli trials each with success probability $0 \le p \le 1$.
The number of quantum measurements or circuit evaluations (per kernel estimation) is denoted by $N$.
The size of the training set is denoted by $m$.
$\psdset$ denotes the space of $m\times m$ positive semidefinite and symmetric matrices. $K$ represents kernel functions and bold font $\K$ represents the kernel matrix over the
training data. $f$ and $h$ denote classifier functions.

\begin{definition}[Subgaussian random variable]
A random variable $X\sim \subg(\sigma^2)$ is subgaussian, if $\EE(X) = 0$ and
\begin{equation}
\EE(e^{sX}) \le e^{\frac{1}{2}s^2\sigma^2} \quad \forall s \in \RR.
\end{equation}
The smallest such $\sigma$ is called the subgaussian norm of $X$.
\end{definition}
The following facts about subgaussian random variables are useful.
\begin{fact}\label{fact:bin}
Define $Z = \frac{1}{N} X - p$, where $X \sim \bin(N,p)$.
Then $Z \sim \subg(\sigma^2)$  for any $\sigma^2 \ge \texttt{Var}(Z)$.
\end{fact}
\begin{fact}
Subgaussian random variables obey the following tail bounds \cite{Buldygin1980}:
\begin{enumerate}
\item If $X \sim \subg(\sigma^2)$ then for every $t > 0$,
$$\prob(X \ge t) \le \exp\bigg(-\frac{t^2}{2\sigma^2}\bigg).$$
\item
\label{fact:tail}
Let $X_{i} \sim \subg(\sigma^2),\ i\in[n]$ be independent random variables.
Then for any $a\in\RR^d$ the random variable $Y=\sum_{i=1}^{n}a_{i}X_{i}$ satisfies
\begin{equation}\label{eq:tail}
\Prob\big(Y \geq t\big)\leq \exp\bigg(-\frac{t^2}{2\sigma^2\|a\|^2}\bigg).
\end{equation}
\end{enumerate}
\end{fact}
This immediately leads to the following assertion.
\begin{lem}
\label{lemma:tail}
Let $Y$ be defined as in Fact ~\ref{fact:tail}.
For any $0<\delta\leq1$,
\begin{equation}
\prob(Y \ge t) \le \delta
\end{equation}
holds whenever
\begin{equation}
t\geq\sigma\norm{a}\kappa(\delta),\nonumber
\end{equation}
where
\begin{equation}
\kappa(\delta) = \sqrt{2\log\bigg(\frac{1}{\delta}\bigg)}.
\end{equation}
\end{lem}

\subsection{Support Vector Machines for Classification}
\label{sec:svm}
A function $f: \X \subset \RR^d \to \{-1,1\}$ is called a binary classifier, which takes an \emph{observation} $\bx$ and outputs a label $y$.
The problem of estimating $f$ from a training dataset
\begin{equation} \label{eq:data}
\D_{\train} = \left\{\left(\bx_{i},y_{i}\right) \left| \bx_i \in \X \subset \RR^d,~ y_i \in \{-1,1\},~ i\in [m]\right.\right\}
\end{equation}
consisting of $m$ i.i.d.\ (\emph{observation},\ \emph{label}) pairs drawn from a distribution $\P$ is of great interest in machine learning.
The aim is to obtain an $f$ such that the generalization error, $\Prob\left(f(\bx) \ne y\right)$, is small for $(\bx,y) \sim \P$.

Support Vector Machines (SVM) are widely used classifiers \cite{Cervantes2020comprehensive}, which have been extremely successful in practical applications.
For a dataset $\D_{\train}$, a SVM kernel classifer is
$$f(\bx| \alpha, b; K ) = \sign\bigg(\sum_{i = 1}^{m} \alpha_i y_i K(\bx_i,\bx) + b\bigg),$$
where the $\alpha =[\alpha_1,\ldots,\alpha_m]^\top$ is obtained by solving the problem
\begin{align*}
\max_{\alpha} &\left(\sum_{i=1}^m \alpha_i - \frac{1}{2} \sum_{i,j} \alpha_i\alpha_jy_iy_j K(\bx_{i},\bx_{j})\right), \tag{\mbox{\texttt{DUAL}}}\label{eq:svm}\\
 \textrm{s.t.} &\sum_{i=1}^{m}\alpha_{i} y_{i} = 0 \ \textrm{and} \ \ 0\leq\alpha_{i}\leq C,
 \end{align*}
for some $C\geq0$.

This optimization problem can be solved efficiently using convex quadratic programming \cite{ScholkopfBook}.
The kernel function $K: \X \times \X \to \RR$ is \emph{positive semidefinite}, and plays the role of a dot product in a suitably defined Reproducing Kernel Hilbert Space (RKHS).
One of the most interesting properties of the kernel function is it implicitly defines an embedding of the observation that does not need to be explicitly computed.
For a comprehensive exposition, see Ref.\ \cite{VapnikBook}.
For brevity we introduce the kernel matrix $\bK \in \psdset$, whose entries are $\K_{ij} = K(\bx_{i},\bx_{j})$ for any $i,j \in [m]$.
Since the kernel function is positive semidefinite, the kernel matrix is also positive semidefinite.

To help develop our ideas, we work with an equivalent setup.
We consider classifiers of the form
\begin{equation}\label{eq:classifier}
f(\bx| \beta, b; K ) = \sign\bigg(\sum_{i = 1}^{m} \beta_i  K(\bx_i,\bx) + b\bigg),
\end{equation}
where $\beta=[\beta_1,\ldots,\beta_m]^\top \in \RR^m,b\in \RR$.
The parameters $\beta,b$ of the classifier are obtained by minimizing
\begin{equation}
\min_{\beta \in \RR^m, b \in \RR} J(\beta, b;\bK),  \tag{\texttt{PRIMAL}} \label{opt:p}
\end{equation}
where
\begin{equation}\label{eq:obj}
\begin{aligned}
J(\beta, b;\bK)=&C \sum_{i=1}^m  \max\left(1 - y_i(\sum_{j=1}^m \beta_j\bK_{ij} + b),0\right) \\
 &+ \frac{1}{2}\sum_{ij}\beta_i\beta_j\bK_{ij}
\end{aligned}
\end{equation}
for some $C \ge 0$ (same as in \ref{eq:svm}) that sets the relative weights of the two terms.
The first term is known as the hinge loss, and the second term is called the regularization term.
Since the kernel function is positive semidefinite, the resulting problem
is thus again a convex quadratic optimization problem.

\subsection{Quantum Embedding Kernels}

Kernel methods have a natural extension into the quantum setting.
A quantum computer can be used to embed the data into a high dimensional Hilbert space and the kernel function can be computed by estimating the state overlaps.
This also offers a path to obtain quantum advantage in machine learning, since certain type of kernel functions are (conjectured to be) hard to evaluate classically but easy to do so quantum mechanically \cite{Havlicek2019,Liu2021rigorous,Kubler2021inductive}.
A quantum embedding circuit takes the classical input datapoint $\bx\in \mathcal{X}\subset \mathbb{R}^{d}$, and maps it to a quantum state $\ket{\phi(\bx)}\in \mathcal{H}$, the {\em computational} Hilbert space.
A valid kernel function $K: \mathcal{X}\times \mathcal{X} \to \mathbb{R}$ can be chosen as
\begin{equation}
\label{eq:QEK}
K(\bx,\bx') \equiv \Tr{\rho(\bx) \rho(\bx')} =  | \braket{\phi(\bx)|\phi(\bx')}|^2,
\end{equation}
where $\rho(\bx) \equiv \ket{\phi(\bx)}\bra{\phi(\bx)}$ is the density matrix that plays the role of a vector in the feature Hilbert space $\mathcal{H}\otimes \mathcal{H}^{*}$.
Such kernels are known as \emph{quantum embedding kernels} (QEK) \cite{Schuld2019,Schuld2021}.

The mapping into the computational Hilbert space $\ket{\phi(\bx)} = U(\bx) \ket{0}^{n}$ is achieved by a unitary operator $U(\bx)$ using a quantum circuit dependent on $\bx$.
For example, we may use the elements of $\bx$ as the angles of rotation in the various
$1$- and $2$-qubit gates used to construct the quantum circuit.
Such an angle embedding circuit is shown in appendix, and is one of the simplest ways to encode the input data $\bx$ into a quantum state $\ket{\phi(\bx)}$.
Given such a circuit, we can map each datapoint to a quantum density matrix, $\rho(\bx) = U(\bx) \ket{0\ldots0}\bra{0\ldots0} U^\dagger(\bx)$.
The kernel can then be expressed as the Hilbert-Schmidt inner product $K(\bx,\bx') = \Tr{\rho{(\bx)}\rho{(\bx')}}$, which can be evaluated using either the SWAP test or the GATES circuit that we describe in the next Section.


\section{Sampling Noise in Quantum Kernels}
\label{sec:q_samp_correction}


Let $K^*:\X\times\X\to\RR$ denote the ideal quantum kernel function and $K^{(N)}$ be the stochastic kernel function estimated from $N$ quantum measurements. The problem of unreliable classification is highlighted and formalized in Sec.\ \ref{sec:prob_statement}.
The quantum circuits for evaluating the kernel function are discussed in Sec.\ \ref{sec:q_circuits}.
We show that these circuits are affected differently by the sampling noise and in Definition \ref{lemma:cf}, we quantify this fact by introducing a circuit factor.
In Theorem \ref{thm:N_bound}, we derive a sufficient condition for $N$ which ensures reliable classification of a single datapoint.

\subsection{Problem of Unreliable Classification}
\label{sec:prob_statement}
\begin{figure}
    \centering
         \centering
         \includegraphics[width=0.5\textwidth]{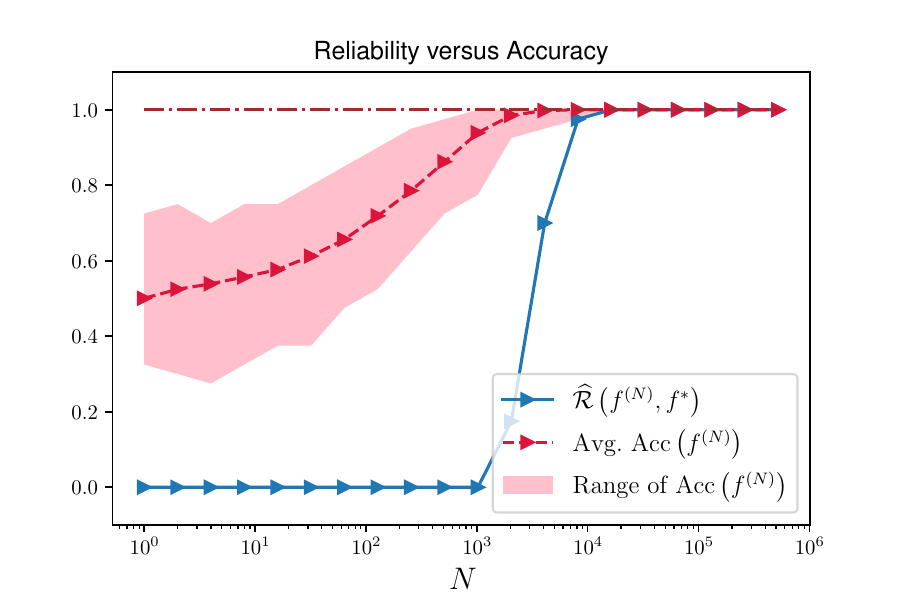}

\captionsetup{justification=raggedright}
\caption{{\bf Reliability and Accuracy.} Performance of \ref{eq:skc*} is shown as a function of $N$, calculated over $N_{\trials}=200$ independent trials. \ref{eq:ekc*} is thus
The observed accuracies $\Acc\left(f^{(N)}\right)$ of \ref{eq:skc*} can show large variations over the different trials
(red shaded area shows the range). Reliability $\widehat{\R}(f^{(N)})$ measures the fraction of data points which are reliably classified (definition \ref{def:emp_rel}).
Notice that even when the average accuracy is high ($>85\%$ at $N=512$), the reliability can be very low ($0\%$ at $N=512$).}
\label{fig:prob_statement}
\end{figure}
We highlight the main problem of quantum kernel classifiers, i.e., they do not reliably reproduce the labels of an ideal quantum classifier even with reasonable values of $N$ (see Fig.\,\ref{fig:prob_statement}). We compare the two classfiers below.
\begin{definition}[{\bf Exact Kernel Classifier}]
\label{def:ekc}
The exact kernel classifier or the true classifier is
\begin{equation}
 f^{*}(\bx) = f\left(\bx|\beta^{*},b^{*};K^{*}\right), \tag{EKC}
\label{eq:ekc*}
\end{equation}
where $\beta^{*},b^{*}$ denote the optimal coefficients minimizing the \ref{opt:p} objective $J(\beta,b;\K^{*})$.
\end{definition}

\begin{definition}[{\bf Stochastic Kernel Classifier}]
\label{def:skc}
Let $N$ denote the number of measurements performed to estimate the kernel function $K(\bx,\bx_{i})$, for each $i\in[m]$. Then
\begin{equation}
f^{(N)}(\bx) = f\left(\bx|\beta^{*},b^{*};K^{(N)}\right), \tag{SKC}
\label{eq:skc*}
\end{equation}
where, as in \ref{eq:ekc*}, $\beta^{*},b^{*}$ denote the optimal coefficients minimizing the \ref{opt:p} objective $J(\beta,b;\K^{*})$.
\end{definition}

Naturally,
\begin{equation}
\lim_{N\to\infty} f^{(N)}(\bx) = f^{*}(\bx).
\end{equation}
But the convergence can be slow in practice, and would depend on the dataset and choice of kernel function $K^{*}$.

Classifier performance over a dataset is generally quantified by its accuracy.
\begin{definition}[{\bf Accuracy}]
\label{def:accuracy}
The accuracy of a classifier $f(\bx)$ over a dataset $\D = \{(\bx_{i},y_{i}),\ i\in[M]\}$,
\begin{equation}
\textrm{Acc}\left(f\right) = \frac{1}{M}\sum_{i=1}^{M}\mathbbm{1}_{f(\bx_{i}) = y_{i}},
\label{eq:acc}
\end{equation}
is the fraction of datapoints classified correctly.
\end{definition}
For a stochastic kernel classifier $ f^{(N)}(\bx)$, the classification of any datapoint $\bx\in\X$ depends upon the evaluation of kernels $K^{(N)}(\bx,\bx_{i}), i\in[m]$
which are random variables. The accuracy (\ref{eq:acc}) is thus also a random variable and may show large variation.

Given some datapoint $\bx \in \X$, we would like the classifier \ref{eq:skc*} to return the same label as \ref{eq:ekc*} with high probability. This is captured by its reliability.

\begin{definition}[{\bf Reliability}]
The reliability of the quantum classifier $f^{(N)}$, at a datapoint and over a dataset, respectively,
are defined as follows.
\label{def:reliability}
\begin{enumerate}
\item {\bf Reliability at a datapoint} $\bx \in \X$ is the probability with which it agrees with the ideal quantum classifier $f^*$.
\begin{equation}
\R(\bx;f^{(N)},f^{*}) \equiv \Prob\left(f^{(N)}(\bx)=f^{*}(\bx)\right).
\label{success_prob}
\end{equation}
\item {\bf Reliability over a dataset} $\D$ of size $M=|\D|$ is the fraction of points in the dataset that
are {\bf reliably} classified
\begin{equation}
\label{rel_dataset}
{\R_{\delta}(f^{(N)},f^*) = \frac{1}{M} \sum_{s=1}^{M}\mathbbm{1}_{{\R}\left(\bx_{s};f^{(N)},f^{*}\right)\geq1-\delta}},
\end{equation}
where $0<\delta\ll1$ denotes a small probability.
\end{enumerate}
\end{definition}
Thus, we say that a point $\bx\in\X$ is {\em reliably classified} if its reliability is very close to unity.
We say that \ref{eq:skc*} {\em reliably reproduces} \ref{eq:ekc*} when {\em all points} in the dataset are reliably classified, i.e.,
with high probability $1-\delta$.

We provide the following empirical metric to quantify the reliability, (\ref{success_prob}) and (\ref{rel_dataset}) respectively, over a number of trials.
\begin{definition}[{\bf Empirical Reliability}]
\label{def:emp_rel}
Let $f^{(N)}{\{1\}}, f^{(N)}{\{2\}}, \ldots, f^{(N)}{\{N_{\trials}\}}$ denote the specific instantiations of a stochastic classifier $f^{(N)}$ over different trials. Then the {\bf empirical reliability}
\begin{enumerate}
\item {\bf at a datapoint} $\bx\in\X$ is the empirical probability that the quantum classifier $f^{(N)}$
agrees with the exact classifier $f^{*}$:
\begin{equation}
\widehat{\R}\left(\bx;f^{(N)},f^{*}\right) = \frac{1}{N_{\trials}}\sum_{k=1}^{N_{\mathrm{trials}}} \mathbbm{1}_{f^{(N)}{\{k\}}(\bx) = f^{{}^*}(\bx)},
\end{equation}
\item {\bf over a dataset} $\mathcal{D}$ is the fraction of points which are classified reliably ($|\D|=M$):
\begin{equation}
\label{reliability}
\widehat{\R_{\delta}}\left(f^{(N)},f^{*}\right) = \frac{1}{M} \sum_{s=1}^{M}\mathbbm{1}_{\widehat{\R}\left(\bx_{s};f^{(N)},f^{*}\right)\geq1-\delta},
\end{equation}
where $0<\delta\ll1$ is a small probability.
\end{enumerate}
\end{definition}
We find that choosing $\delta=0$ hardly affects our results for empirical reliability. With this choice, we drop the subscript $\delta$ henceforth from the empirical reliability (\ref{reliability}).
In our numerical experiments, we set $N_{\trials}=200$.

{\bf Problem of Unreliable Classification:}
An illustration of our numerical experiments is shown in Fig.\ \ref{fig:prob_statement}.
We observe that a large number of measurements are needed for \ref{eq:skc*} to
reliably reproduce the labels {\em even for the training data}.
Notice:
\begin{enumerate}
\item A high accuracy on {\em average} carries little information regarding how {\em reliably} an observation $\bx$ is classified by \ref{eq:skc*}.
For example, the average accuracy $>85\%$ for $N=512$, but {\em none} of the points are classified reliably.

\item A wide range of the observed accuracies over different instantiations of \ref{eq:skc*} indicates unreliable classification.
\end{enumerate}
As a function of $N$, $\widehat{\R}(f^{(N)},f^{*})$ is $0$ for a large range of $N$ values where the average accuracies can be high in Fig.\ \ref{fig:prob_statement}. In this sense, the reliability is indicative of the worst instantiation of SKC: if the worst performing instantiation of SKC is able to reproduce the output of \ref{eq:ekc*} over a dataset, then we can say that SKC is reliable.

\subsection{Quantum sampling noise is circuit-dependent}
\label{sec:q_circuits}
We shall use the shorthand notation $K^{*}_{ij} = K^{*}(\bx_{i},\bx_{j})$ and $K^{(N)}_{ij} = K^{(N)}(\bx_{i},\bx_{j})$. We denote their difference at finite sampling as
\begin{equation}\label{DeltaK}
\Delta{K}^{(N)} _{ij} = K^{*}_{ij} - K^{(N)}_{ij}, \ \textrm{for any}\ \bx_i,\bx_j\in\X.
\end{equation}
In what follows, we demonstrate that the sampling noise depends on the quantum logic circuit used to evaluate the kernel function.

We consider the two different ways in which the overlaps $K_{ij} = |\bra{\phi(\bx_{i})}\phi(\bx_{j})\rangle|^{2}$ can be evaluated \cite{Havlicek2019}.
The data $\bx_{i}\in\mathcal{X}$ is transformed into a state in the computational Hilbert space $\mathcal{H}$ by starting with all $n$ qubits in the $\ket{0}$ state and using the unitary operator
\begin{equation}
\ket{\phi(\bx_{i})} = U(\bx_{i}) \ket{0}^{\otimes n}.
\end{equation}
The first method is to express the kernel as $K(\bx_{i},\bx_{j}) = |\bra{0}^{\otimes n} U^{\dagger}(\bx_{i})U(\bx_{j})\ket{0}^{\otimes n}|^{2}$, which can be evaluated by applying $U^{\dagger}(\bx_{i})U(\bx_{j})$ to the initial state $\ket{0}^{\otimes n}$ and then estimating the probability that the final state is also $\ket{0}^{\otimes n}$.
We refer to this as the GATES test, shown in fig.\ \ref{fig:circuit}(a).
It has also been called the inversion test \cite{Lloyd2020} and the adjoint test \cite{Hubregtsen2021}.
The value of the kernel function  $K^{(N)}_{ij} = K^{(N)}(\bx_{i},\bx_{j})$ is inferred
using $N$ independent Bernoulli trials, each of which gives the final state $\ket{0}^{\otimes n}$ with probability $p = K^*_{ij}$.
Hence,
\begin{equation}\tag{GATES}
K^{(N)}_{ij} \sim \frac{1}{N}\ \bin\bigg(N, p = K^{*}_{ij}\bigg).
\label{gates_dist}
\end{equation}

An alternate way to measure fidelity or state overlap is the SWAP test \cite{NielsenChuangBook}, shown in fig.\ \ref{fig:circuit}(b).
In this case, the ancilla bit is read out, and the probability of obtaining the state $\ket{0}$ can be written in terms of the state overlap
\begin{align}
\Prob\big(\textrm{ancilla} = \ket{0}\big) =& \frac{1}{2}\big(1 + |\bra{\phi(\bx_{i})}\phi(\bx_{j})\rangle|^{2}\big)\\
 = &\frac{1}{2} \left(1 + K_{ij}^{*}\right).
\end{align}
We infer $K_{ij}^{(N)}$ through $N$ independent Bernoulli trials that measure the state of the ancilla bit, and therefore obtain
\begin{equation}\tag{SWAP}
K^{(N)}_{ij} \sim \frac{2}{N}\ \bin\bigg(N, p =\frac{1}{2} (1+ K^*_{ij})\bigg) - 1.
\label{swap_dist}
\end{equation}

It is obvious that in both (\ref{gates_dist}) and (\ref{swap_dist}) cases
$K^{(N)}_{ij}$ is an \emph{unbiased estimator} of $K^*_{ij}$,
\begin{equation}
\EE[K^{(N)}_{ij}] = K^{*}_{ij}\ \ \forall \bx_{i},\bx_{j} \in \X.
\end{equation}
The variance of the estimator, however, is linked to the specific circuit and the number of measurements $N$.
\begin{figure}
     \centering
     \begin{subfigure}[hb]{0.5\textwidth}
         \centering
         \includegraphics[width=\textwidth, trim = 2.2cm 24cm 10cm 5cm]{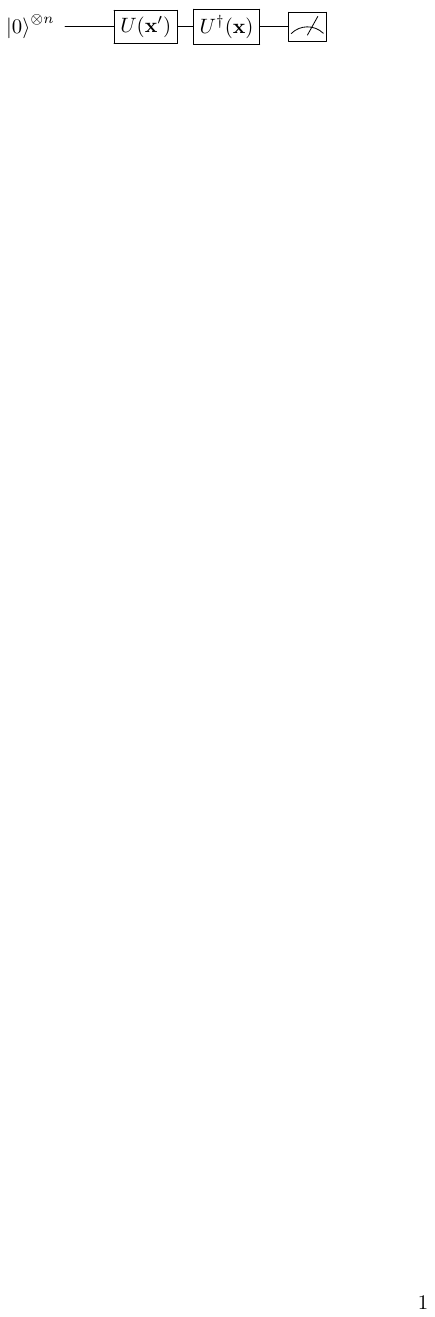}
         \caption{}
         \label{fig:gates}
     \end{subfigure}

\begin{subfigure}[hb]{0.5\textwidth}
         \centering
         \includegraphics[width=\textwidth, trim = 2.2cm 22cm 10cm 4cm]{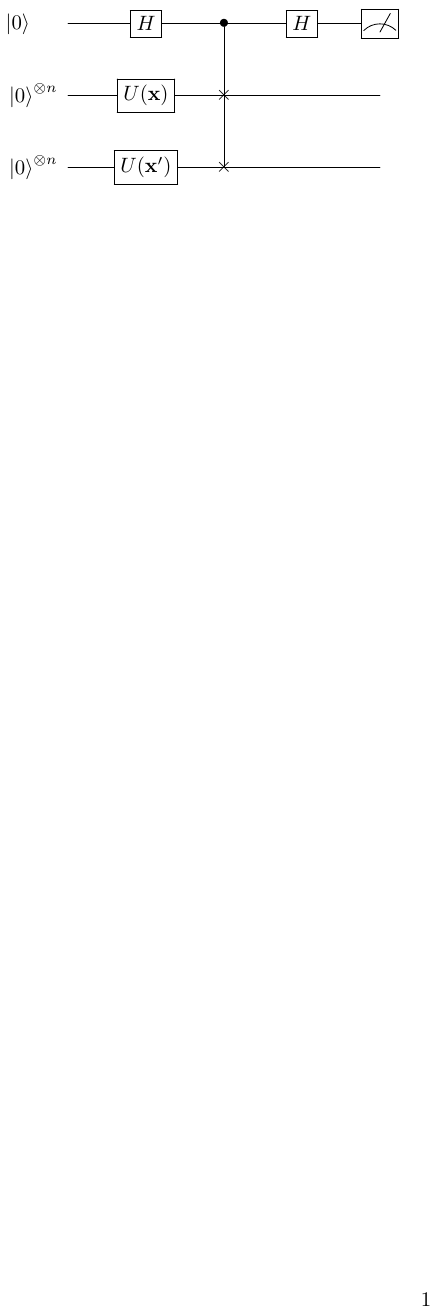}
         \caption{}
         \label{fig:swap}
     \end{subfigure}
\captionsetup{justification=raggedright}
\caption{(a) GATES circuit measures the probability of obtaining the final state $\ket{0}^{\otimes n}$. The resulting distribution for $K(\bx,\bx')$ is given by (\ref{gates_dist}).
(b) SWAP circuit measures the probability of obtaining the ancillary qubit (top line) in the $\ket{0}$ state. The resulting kernel estimate is given by (\ref{swap_dist}).}
\label{fig:circuit}
\end{figure}

\begin{definition}[Circuit Factor]
\label{lemma:cf}
The variance of estimator $K_{ij}^{(N)}$ is upper-bounded by $\sigma_0^2$, where
\begin{equation}
\label{eq:cf}
\sigma_0 = \frac{\c}{2\sqrt{N}}, \quad
\begin{aligned}
\c &= \begin{cases}
 1 &  (\textbf{GATES}),\\
 2 & (\textbf{SWAP}).
    \end{cases}
\end{aligned}
\end{equation}
The {\em circuit factor} $\c$ is a circuit-dependent constant.
\end{definition}
For a derivation of the $\sigma_0$ value, see Appendix \ref{sec:subgaussian}.
For both the circuits under consideration, $K^{(N)}_{ij}$ is an affine function of a {\tt Binomial} random variable.
Direct application of Fact\ \ref{fact:bin} then yields
\begin{equation}
\label{eq:um}
\tag{\texttt{UM}}
\Delta K^{(N)}_{ij} \sim \subg(\sigma_0^2).
\end{equation}

We study the problem of kernel based classification under this uncertainty model, and only consider the fundamental sampling noise.
We however note the following extension in presence of other sources of noise, and discuss the modification to the \ref{eq:um} in Appendix \ref{sec:datasets}.
\begin{remark}
\label{rem:gen_um}
The uncertainty model \ref{eq:um} subsumes all other sources of unbiased noise.
\end{remark}
\begin{proof}
Since all quantum kernel distributions $0\leq K(\bx,\bx')\leq1$ have finite support, they are subgaussian due to Hoeffding's lemma (e.g.\ see 2.6 in \cite{Massart2007}).
\end{proof}

We note that the kernel estimate given by (\ref{gates_dist}) has a lower variance compared to the estimate given by (\ref{swap_dist}).
This implies that, purely from the perspective of sampling noise, the GATES test is preferable for any $N$.
The GATES test requires a circuit whose depth is twice that of the SWAP test, but the SWAP test requires twice the number of qubits as the GATES test and additional controlled SWAP operations.
The practical choice of the test to use for estimating QEKs can therefore depend on other factors like the coherence times and the quality of the gates available.


\subsection{Ensuring a Reliable Classification}

When classifying a datapoint $\bx\in\X$ using quantum kernels, we seek $N$ that gives high reliability ($\delta\ll 1$) in Eq.\ (\ref{success_prob}). This immediately leads to the following lower bound on $N$.
\begin{thm}
\label{thm:N_bound}
Let $\gamma(\bx) = y\left(\sum_{j}\beta^{*}_{j} K^{*}(\bx,\bx_{j}) + b^*\right)$ represent the margin of classification of the ideal classifier \ref{eq:ekc*} for the point $(\bx,y)$. Then \ref{eq:skc*} has a reliability of at least $1-\delta$, whenever
\begin{equation}
N\geq \frac{1}{2}\left(\frac{\c \norm{\beta^{*}}}{\gamma(\bx)}\right)^{2} \log\left(\frac{1}{\delta}\right),
\end{equation}
where $\c$ is the circuit factor.
\end{thm}
\begin{proof}
This bound is a direct application of Lemma \ref{lemma:tail}, along with the uncertainty model \ref{eq:um}.
\end{proof}
Theorem \ref{thm:N_bound} immediately implies that \ref{eq:skc*} can
reliably classify any point: $\R(\bx;f^{(N)},f^{*})= 1-\delta$ can be made arbitrarily close to 1 by increasing $N$ by at most a logarithmic factor $\log(1/\delta)$.

In NISQ devices, the effect of noise on the results is extremely important and has received wide attention \cite{Bharti2022noisy}.
Recent works in quantum algorithms thus place special emphasis on the sampling noise, and treat the number of measurements $N$ as a resource that should be used minimally \cite{Kubler2020frugal,Arrasmith2020frugal,Gu2021}.
Our observations from Fig.\ \ref{fig:prob_statement} and the constraints of near-term devices raise the following questions:
\begin{enumerate}
\item Can we derive a generic bound for $N$ that would ensure {\em reliable} classification over any dataset?
\item Can we construct a reliable classifier which uses fewer $N$?
\end{enumerate}
We now answer these questions in Sections \ref{sec:expt_1} and \ref{sec:expt_2} respectively.

\section{Bounds on the Number of Measurements}
\label{sec:expt_1}
We here derive bounds for the number of measurements $N$ which can ensure a small classification error, and perform numerical experiments showing the validity and usefulness of the bounds.
In Theorem \ref{thm:emp_risk}, we first derive a theoretical bound $N_{\theory}$ which ensures, with high probability, that the stochastic kernel classifier \ref{eq:skc*} has a classification error smaller than the margin error of the exact kernel classifier \ref{eq:ekc*}. We then provide an empirical bound $N_{\practical}$ that is tighter and discuss their relationship.

\begin{definition}[{\bf $\gamma$-margin error}]
\label{margin_violation}
Given a dataset $\D\sim\P$, and a margin $\gamma\geq0$, the set of points in $\D$ which are $\gamma$-margin-misclassified by the classifier $f(\bx) = f(\bx|\beta,b;K)$ (\ref{eq:classifier}) is
\begin{equation}
\myS_{\gamma}\left[f\right] =  \left\{(\bx_i,y_i) \in \mathcal{D}~\left|~y_i\bigg(\sum_{j = 1}^{m} \beta_j K(\bx_j,\bx_{i}) + b\bigg)  < \gamma  \right.~ \right\}.
\nonumber
\end{equation}
The $\gamma$-margin error of the classifier $f$ over the dataset $\D$ is defined as
\begin{equation}
\label{eq:epsilon}
\epsilon_{\gamma}(f) = \frac{\big|\myS_{\gamma}[f]\big|}{|\D|}.
\end{equation}
\end{definition}

With $\gamma =0$ in the above expression, the set $\myS_{0}$ includes all the misclassifications made by the classifier over the dataset $\mathcal{D}$.
Eq.\ (\ref{eq:epsilon}) thus gives the \emph{classification error}
\begin{equation}
\epsilon_{0}\left(f\right) = 1- \textrm{Acc}\left(f\right)
\end{equation}
of the classifier over dataset $\D$.
On the other hand if $\gamma = 1$, the set $\myS_{1}$ includes all those points $(\bx_{i},y_{i})$, which lie on the wrong side of the two margins $y_{i}\left(\sum_{j = 1}^{m} \beta_j K(\bx_j,\bx_{i}) + b\right) = 1$ that define the standard \emph{margin error} $\epsilon_{1}(f)$ of the SVM classifier.
These are the points with nonzero contributions to the hinge loss term in (\ref{eq:obj}).

The main result of this Section is that the classification error $\epsilon_{0}\left({f}^{(N)}\right)$ of \ref{eq:skc*} can be guaranteed,
with high probability, to be smaller than the $\gamma$-margin error $\epsilon_{\gamma}\left(f^{*}\right)$ of \ref{eq:ekc*}, when
$N$ is sufficiently large.

\begin{thm}
Let $\mathcal{D}\sim \P$ denote a dataset of size $M=|{\mathcal{D}}|$.
The classification error of the stochastic kernel classifier $f^{(N)}(\bx)= f\left(\bx|\beta^*,b^*;K^{(N)}\right)$ is bounded by the $\gamma$-margin error of the true classifier $f^{*}(\bx)=f\left(\bx|\beta^{*},b^{*};K^{*}\right)$,
\begin{equation}
\label{eq:margin_bound}
\epsilon_{0}\left({f}^{(N)}\right) \le  {\epsilon}_{\gamma} \left(f^*\right),
\end{equation}
with probability at least $1-\delta$, whenever
\begin{equation}
\label{eq:N_bound}
N \geq  \frac{\c^{2}}{2\gamma^2}\norm{\beta^{*}}^{2} \log \frac{M}{\delta}.
\end{equation}
\label{thm:emp_risk}
\end{thm}

\begin{figure*}
\centering
      \begin{subfigure}[hb]{0.48\textwidth}
         \centering
         \includegraphics[width=\textwidth,trim = 0cm 0cm 0cm 0cm]{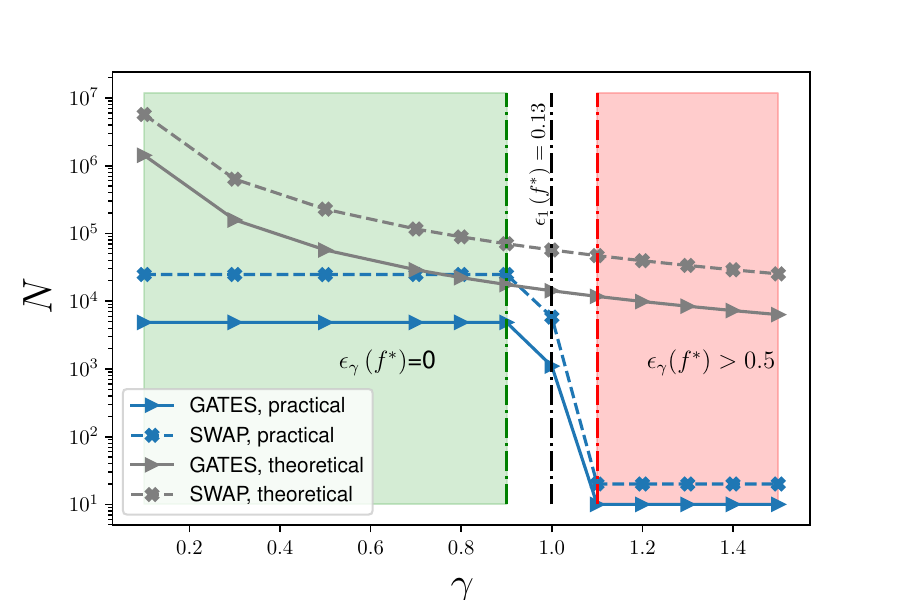}
          \caption{}
                    \label{fig:N_prac}

      \end{subfigure}
\hspace{-0.05\textwidth}
      \begin{subfigure}[hb]{0.48\textwidth}
         \centering
         \includegraphics[width=\textwidth,trim = 0cm 0cm 0cm 0cm]{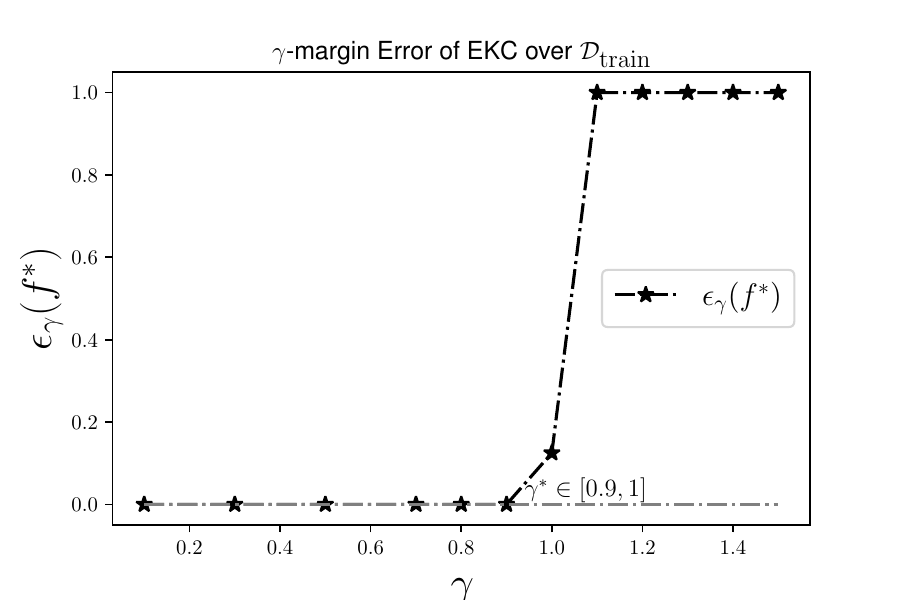}
         \caption{}
              \label{fig:margin_error}

      \end{subfigure}

\captionsetup{justification = raggedright}
\caption{(a) Practical and theoretical bounds on $N$ as a function of the margin $\gamma$.
The green shading represents the region $\epsilon_{\gamma}\left(f^{*}\right)=0$ where the stochastic kernel function \ref{eq:skc*} is expected to have perfect accuracy, with $\delta_{\target}=0.01$.
The red shading represents the region with large margin error, $\epsilon_{\gamma}\left(f^{*}\right)>0.5$, and low accuracy.
The theoretical bound is smooth as a function of $\gamma$, but the practical bound shows a rapid change in $N$ around $\gamma=1$.
The positions of the vertical green and red boundaries are expected to be problem dependent. The theoretical bound on $N$
to obtain an accurate classifier \ref{eq:skc*} is optimal when we choose the largest $\gamma$ in the green region.
(b) The $\gamma$-margin error of the Exact Kernel Classifier. $\epsilon_{\gamma}(f^{*})=0$ for $\gamma>0$
implies that \ref{eq:ekc*} makes no errors over the dataset, i.e., $\epsilon_{0}(f^{*})=0$. It can be seen that $\gamma^{*} = \argmax_{\gamma} \{\epsilon_{\gamma}(f^{*})=0\}$ lies in the interval $[0.9,1]$.}
\end{figure*}

This theorem essentially extends the bound of Theorem\ \ref{thm:N_bound} to any dataset of size $M$.
Theorem \ref{thm:emp_risk} can be seen as a statement about the reliability of \ref{eq:skc*} over a dataset:
\begin{cor}
\label{cor:Reliability}
Let \ref{eq:ekc*} have an accuracy of $1-\epsilon_{0}$
over dataset $\D$ of size $M$ (we denoted $\epsilon_{0}\left(f^*\right)=\epsilon_{0}$). Let
\begin{equation}
\label{eq:max_gamma}
\gamma^{*} = \argmax_{\gamma} \{\epsilon_{\gamma}\left(f^{*}\right)= \epsilon_{0}\}.
\end{equation}
Then \ref{eq:skc*} reliably reproduces \ref{eq:ekc*} over $\D$
\begin{equation}
\R_{\delta}\left(f^{(N)},f^{*}\right) = 1,
\end{equation}
i.e., each datapoint in $\D$ has been classified with a reliability of at least $1-\delta$ whenever
\begin{equation}
N\geq \frac{\c^{2}}{2{\gamma^{*}}^2}\norm{\beta^{*}}^{2} \log \frac{1}{\delta}.
\end{equation}
\end{cor}

We make an important remark regarding Theorem \ref{thm:emp_risk}.
\begin{remark}
Given any dataset $\D\sim\P$ of size $M$, the number of measurements,
\begin{equation}
N = \Big{\lceil}{\frac{\c^{2}}{2\gamma^2} m_{\sv}C^{2} \log \frac{M}{\delta}}\Big{\rceil},
\end{equation}
is sufficient to ensure that (\ref{eq:margin_bound}) holds with probability at least $1-\delta$.
Here $m_{\sv}$ is the number of support vectors in the training set, and $C$ is the SVM regularization parameter \ref{opt:p}.
\label{remark:emp_risk}
\end{remark}
\begin{proof}
The optimal primal and dual coefficients are related by $\beta^{*}_{i} = y_{i}\alpha^{*}_{i}$ \cite{ScholkopfBook}. Since $0\leq \alpha_{i}\leq C$, we have $\norm{\beta^{*}}^{2} \leq m_{\sv}C^{2}$, and the proof follows.
\end{proof}
This remark provides guidance for choosing the kernel function.
A good choice of kernel $K^{*}$ should classify accurately with a large
margin $\gamma$, a small regularization parameter $C$, and few support vectors $m_{\sv}$.

\begin{definition}[{\bf Theoretical and practical bounds on $N$}]
\label{def:N_prac}
Let the margin $\gamma\geq0$ be specified for a dataset $\D\sim\P$ of size $M=|\D|$.
\begin{enumerate}
\item We define the theoretical bound as the subgaussian bound derived in (\ref{eq:N_bound}),
\begin{equation}
\label{eq:N_theo}
N_{\theory} = \Big\lceil\frac{\c^{2}}{2\gamma^2}\norm{\beta^{*}}^{2} \log \frac{M}{\delta_{\target}}\Big\rceil,
\end{equation}
where $\delta_{\target}\ll1$ is set to a small value.
\item We find the practical bound $N_{\textrm{practical}}$ by calculating the empirical probability of satisfying (\ref{eq:margin_bound}),
\begin{equation}
\label{eq:emp_prob}
1-\delta_{\emp} = \frac{1}{N_\trials} \sum_{s=1}^{N_\trials} \mathbbm{1}_{\epsilon_{0}\left(f^{(N)}\{s\}\right) \leq \epsilon_{\gamma}\left(f^{*}\right)},
\end{equation}
starting with some small $N$ and increasing it to $N_{\practical}$ when we obtain $\delta_{\emp}<\delta_{\target}$.
An explicit algorithm is given in Appendix \ref{alg:N_prac}.
\end{enumerate}
\end{definition}

\begin{table}
  \centering
  \renewcommand{\arraystretch}{1.2}
\begin{tabular}{|p{1.9cm}|p{1.15cm}|p{1.62cm}|p{1.15cm}|p{1.3cm}|p{0.73cm}|}
\hline
\multirow{3}{*}{\thead{Dataset}} & \multicolumn{4}{c|}{$N$ needed for $\epsilon_{0}\left(f^{(N)}\right)=0$, $\delta_{\target}=0.01$} & %
     \multirow{3}{*}{\thead{$\epsilon_{1}(f^{*})$}}\\
\cline{2-5}
 & \multicolumn{2}{c|}{\textbf{GATES}} & \multicolumn{2}{c|}{\textbf{SWAP}} & \\
\cline{2-5}
 & {theory} & {\textbf{practical}} & {{theory}} & {\textbf{practical}} & \\
\hline
Circles & $1.8\times 10^4$ & $\mathbf{4.9\times10^3}$ $\ddagger$ & {$7\times10^4$} & $\mathbf{2.5\times 10^4\ddagger}$ & 0.13 \\
\hline
Havlicek & 330 & {\bf 110} & 1310 & {\bf 540} & 0.13 \\
\hline
Two Moons & $2\times10^4$ & $\mathbf{2.5\times10^3}$ & $7.8\times10^4$ & $\mathbf{9.7\times10^{3}\ddagger}$ & 0.14 \\
\hline
Checkerboard & $4.5\times10^4$ & $\mathbf{7.1\times10^3\ddagger}$ & $1.8\times10^5$ & $\mathbf{3.3\times10^{4}\ddagger}$ & 0.12 \\
\hline
\end{tabular}

\captionsetup{justification=raggedright}
\caption{The number of measurements $N$ needed to ensure that \ref{eq:skc*} {\em reliably} reproduces \ref{eq:ekc*} over the training set itself, with probability of at least $99\%$.
These numbers correspond to the green vertical line in Fig.\ \ref{fig:N_prac}.
The practical bound on $N$ is significantly smaller than the bound implied by Theorem \ref{thm:emp_risk}.
The GATES test requires an $N$ which is about 4 times smaller than the SWAP test, as implied by the circuit factor \ref{lemma:cf}.
The right-most column shows the standard margin error for \ref{eq:ekc*}.
$\ddagger$ indicates that $N$ exceeds the IBMQ default value $4000$.}
\label{tab:N_prac}
\end{table}

\begin{cor}
\label{cor:max_gamma}
Over any dataset $\D$ of size $M$, let $$\gamma^{*} = \argmax_{\gamma} \{\epsilon_{\gamma}(f^{*})=0\}.$$
If $\gamma^*>0$ exists, it
denotes the maximum margin for which the \ref{eq:ekc*} classifies all datapoints correctly. Then
the \ref{eq:skc*} also classifies all the datapoints correctly, with probability $1-\delta_{\target}$, whenever
$$N\geq N_{\theory}=\Big\lceil\frac{\c^{2}}{2{\gamma^{*}}^2}\norm{\beta^{*}}^{2} \log \frac{M}{\delta_{\target}}\Big\rceil.$$
\end{cor}
We are interested in obtaining reliable reproduction of \ref{eq:ekc*} from \ref{eq:skc*}.
Note that $\epsilon_{0}(f^{*})=0$ for all the datasets we have considered.
Figure \ref{fig:N_prac} shows the theoretical and practical bounds on $N$ as a function of the margin $\gamma$ for the Circles dataset.
The green shading in it represents the region where $\epsilon_{0}\left(f^{(N)}\right)=0$ is satisfied with $\delta_{\target}=0.01$.
It thus leads to SKC which reliably reproduces EKC.
In this region, the theoretical bound grows as $\gamma\to0$ with an inverse quadratic dependence. The practical bound is insensitive to $\gamma$, since it only
depends on the value of the $\gamma$-margin error $\epsilon_{\gamma}(f^{*})$ and not the margin $\gamma$ itself.
The two bounds are closest to each other at the right edge of the green region, and these values have been tabulated in Table \ref{tab:N_prac}.
We notice that the practical bounds in the table are about $3-5$ times smaller, for both circuits.
This looseness of the theoretical bound is expected, since we used a subgaussian bound for the binomial distribution and not its complete information.
The GATES circuit provides bounds (both theoretical and practical) which are about $4$ times smaller than those for the SWAP circuit, consistent with the Definition \ref{lemma:cf}.

Very importantly, the bound $N_{\theory}$, needed for reliably reproducing \ref{eq:ekc*}, is logarithmic in the size $M$ of the dataset and linear in the number of support vectors $m_{\sv}$ of the training set.
Often the number of support vectors is significantly smaller than the size $m$ of the training set, $m_{\sv}\ll m$ \cite{Steinwart2003sparseness}. Furthermore, $\gamma^{*}$ in Corollary \ref{cor:Reliability} is not vanishingly small and can be close to unity for a well-chosen kernel function. This constitutes a dramatic improvement in the scaling for $N$ that would otherwise hold for kernel estimations with specified precision $\varepsilon$ and is detailed in Appendix \ref{sec:precise}.

We now show that it is possible to do better than $N_{\practical}$ by constructing a Robust Classifier using chance constraint programming.

\section{Chance Constraint Programming}
\label{sec:expt_2}

The optimization problem over training data is a convex quadratic program since $\K^*\in\psdset$.
In reality, the kernel matrix over a training dataset suffers from the same underlying noise which affects the classifier \ref{eq:skc*}.
If we measure the kernel matrix $\K^{(N)}$ with $N$ measurements per kernel entry, the resulting optimization problem \ref{opt:p} would generally not be convex.
Such situations are encountered in optimization problems and can be handled by introducing a chance constraint \cite{Charnes1958,Prekopa1970}.
We show here how we may use the noise information (\ref{eq:um}) of the training kernel matrix to construct more reliable classifiers.

An optimization problem with stochastic terms can be turned into a chance constraint problem as follows.
Let $\K$ denote the stochastic training kernel matrix.
\begin{equation}
\label{eq:c}
\begin{aligned}
&\min_{\beta \in \RR^m, b,\tau \in \RR} \quad \tau\ \\
\textrm{s.t.\ }\Prob&\left(J(\beta, b;\bK) \le \tau\right) \ge 1 - \delta ,
\end{aligned}
\end{equation}
where $\delta \ll1$ ensures that for most instantiations of the stochastic kernel matrix $\K$, i.e.\ with high probability,
the primal objective $J\left(\beta,b;\bK\right)$ is smaller than $\tau$.
This $\tau$ is the new objective which we wish to minimize.
It is in fact a dummy variable that is the means to minimize the (stochastic) primal objective with high probability.
The feasible set of a chance constraint is generally non-convex and makes the optimization problem intractable (although it is convex for some families of distributions \cite{PrekopaBook}).
The natural strategy in such cases is to construct convex approximations to the chance constraint \cite{Nemirovski2007}.

{\em Relation to previous work.} The chance constraint method has been applied to the dual SVM formulation by Bhadra et.\ al.\ \cite{Bhadra2010}, where convex approximations were derived for the chance constraint for a few different noise distributions.
In the present article, we derive a Robust Program using the chance constraint approach in the primal SVM formulation.
The \ref{opt:p} SVM objective has two terms which are stochastic, namely, the hinge loss term and the regularization term, given in the two lines of Eq.\ (\ref{eq:obj}) respectively. This SVM formulation is better suited to derive
a bound on $N$ since the margin constraints appear explicitly (see Definition \ref{margin_violation}). The hinge loss term in the \ref{opt:p} objective implicitly contain these constraints as shown below.
\subsection{Chance Constraint Approach to the SVM Primal}
\label{sec:chance}
We first rewrite the SVM Primal objective by introducing two dummy variables
(a) slack variables $\xi$ and (b) a variable $t$ for the quadratic term in the objective function
\begin{equation}
\label{eq:obj2}
\min_{\xi,t,\beta,b} J(\xi,t,\beta,b) := C\sum_{i}\xi_{i} + t
\end{equation}
subject to the constraints $\xi_{i}\geq0\ \ \ \forall i\in[m]$, and also
\begin{equation}\tag{constraint-1}
\label{eq:c1}
\xi_{i}\geq 1 - y_{i}\left(\sum_{j}\K_{ij}\beta_{j} +b \right),
\end{equation}
together with
\begin{equation}\tag{constraint-2}
\label{eq:c2}
t\geq \frac{1}{2} \beta^{\T}\K\beta.
\end{equation}
It is easy to see that the above formulation is completely equivalent to \ref{opt:p}.
The two constraints are dependent on the training kernel matrix $\K$ which
is stochastic. We shall ensure that they hold true with high probabilities $1- \delta_{1}$ and $1-\delta_{2}$
respectively. Note that \ref{eq:c1} is a set of $m$ individual constraints and we shall invoke the union bound
and ensure that {\em all} the constraints are satisfied with a probability of at least $1- \delta_{1}$.
\begin{equation}\tag{chance-1}
\label{eq:cc1}
\Prob\bigg(\xi_{i}\geq 1 - y_{i}\big(\sum_{j}\K_{ij}\beta_{j} +b \big)\bigg)\geq1-\frac{\delta_{1}}{m}.
\end{equation}
Similarly, we impose
\begin{equation}\tag{chance-2}
\label{eq:cc2}
\Prob\left(t\geq \frac{1}{2} \beta^{\T}\K\beta\right)\geq 1-\delta_{2}
\end{equation}
We now derive convex approximations to the chance constraints above. The resulting convex program would then be robust to the sampling noise by construction.
\begin{thm}(Shot-frugal and Robust Program)
\label{thm:rob}
For a stochastic training kernel matrix $\K^{(N)}\in\RR^{m}\times\RR^{m}$ inferred from $N$ measurements per entry,
the optimal \ref{opt:p} objective $J^{*}\big(\K^{(N)}\big) := \min_{\beta,b} J\big(\beta,b;\K^{(N)}\big)$ is upper bounded by
\begin{equation}
\label{eq:whp_rob}
J^{*}\big(\bK^{(N)}\big)\leq J^{*}_{\rob},
\end{equation}
with probability at least $1-\delta_{1}-\delta_{2}$. Here
\begin{equation} \label{eq:rob_primal}
J^{*}_{\rob} = \min_{\beta,b,\xi,t} J_{\rob}\left(\beta,b,\xi,t; \bK^*| N, \delta_{1},\delta_{2} \right) \tag{\tt{ShofaR}}
\end{equation}
is the optimal value of a Second-Order Cone Program (SOCP), whose objective is given by
\begin{equation}
\label{eq:J_rob}
J_{\rob}\left(\beta,b,\xi,t; \bK^*| N,\delta_{1},\delta_{2}\right) := C \sum_{i}\xi_{i}+t
\end{equation}
subject to the constraints $\xi_{i}\geq0 \ \ \ \forall i \in[m]$, and also
\begin{equation}\tag{suff-chance-1}
\label{eq:sc1}
\xi_{i}\geq 1 + \frac{\c}{2\sqrt{N}} \kappa{\bigg(\frac{\delta_{1}}{m}\bigg)}\norm{\beta} - y_i(\sum_{j=1}^m \beta_{j}\bK^{*}_{ij} + b),
\end{equation}
together with
\begin{equation}\tag{suff-chance-2}
\label{eq:sc2}
t\geq\frac{1}{2} \beta^{\top}\left(\bK^* + \frac{\c}{2\sqrt{N}}\kappa(\delta_{2})  \mathbf{I}_{m\times m}\right)\beta.
\end{equation}

\end{thm}

A detailed proof of the above theorem is given in Appendix \ref{sec:proofs}. Note that the \ref{eq:rob_primal} objective is very similar to the \ref{opt:p} objective with two key differences:
\begin{enumerate}
\item The hinge-loss has an additional term whose origin is precisely the enforcement of $y_{i}\big(\sum_{j}\beta_{j}\K^{(N)}_{ij} + b\big) \geq 1$, with high probability, for all the $m$ datapoints.
Thus a condition against margin violation, which is {\em robust to stochasticity in the kernel matrix}, is built into the optimization problem \ref{eq:rob_primal} itself. (\ref{eq:sc1}) represents the sufficient
condition for which the chance constraint (\ref{eq:cc1}) is satisfied.
\item The kernel matrix in the regularization term is made larger by adding to its diagonal entries a term which arises precisely by requiring $\beta^{\top}\K^{(N)}\beta\leq\beta^{\top}\K^{*}\beta+t$ with probability $1-\delta_{2}$ for some $t>0$.
The large deviation bound given by Lemma \ref{lemma:tail} implies that $t\geq \frac{\c}{2\sqrt{N}}\norm{\beta}^{2}\kappa(\delta_{2})$ is sufficient. (\ref{eq:sc2}) represents the sufficient
condition for which the chance constraint (\ref{eq:cc2}) is satisfied.
\end{enumerate}

The points above suggest that the classifier constructed by solving \ref{eq:rob_primal} would be more robust to the stochasticity of the training kernel matrix.
We thus expect it to be a more reliable classifier at least over the training dataset.
Moreover, since the kernel evaluations in a test dataset suffer from the same underlying noise, we must expect it to be more reliable even over an independent test dataset (these results are shown in the Appendix \ref{sec:datasets}).
\begin{definition}
The robust stochastic kernel classifier is
\begin{equation}
h^{(N)}(\bx) = f\left(\bx|\beta^{*}_{\rob},b^{*}_{\rob};K^{(N)}\right), \tag{RSKC}
\label{eq:rskc*}
\end{equation}
where $\beta^{*}_{\rob},b^{*}_{\rob}$ are solutions to the optimization problem \ref{eq:rob_primal} for some fixed $N$, and $\delta_{1},\delta_{2}\ll1$.
\end{definition}
\begin{definition}
The robust exact kernel classifier is
\begin{equation}
h^{*}(\bx) = f\left(\bx|\beta^{*}_{\rob},b^{*}_{\rob};K^{*}\right), \tag{REKC}
\label{eq:rekc*}
\end{equation}
where $\beta^{*}_{\rob},b^{*}_{\rob}$ are solutions to the optimization problem \ref{eq:rob_primal}. $N$, and $\delta_{1},\delta_{2}\ll1$ are the same as in \ref{eq:rskc*}.
\end{definition}
Note that the classifier $h^{*}$ is dependent on $N,\delta_{1},\delta_{2}$ via the optimization problem \ref{eq:rob_primal}, but it is {\em not stochastic}.
\ref{eq:rekc*} employs the exact kernel function $K^{*}$ just as the true classifier \ref{eq:ekc*}.
\begin{figure}
\centering
    \begin{subfigure}[hb]{0.48\textwidth}
    \centering
    \includegraphics[width=\textwidth,trim = 0cm 0cm 0cm 0cm]{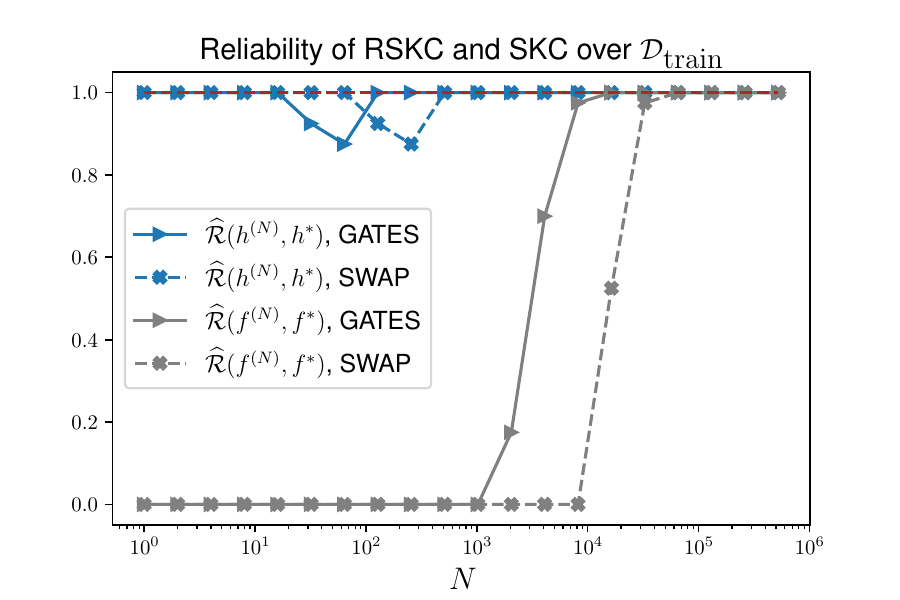}
    \caption{}
    \label{fig:rel_compare}
    \end{subfigure}

    \begin{subfigure}[hb]{0.48\textwidth}
    \centering
    \includegraphics[width=\textwidth,trim = 0cm 0cm 0cm 0cm]{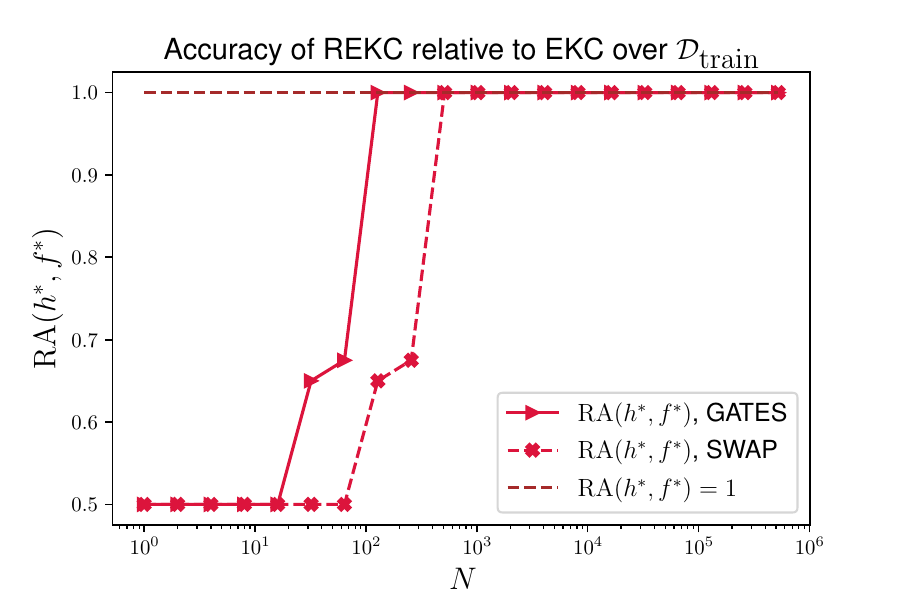}
    \caption{}
    \label{fig:rel_acc}
    \end{subfigure}

\captionsetup{justification=raggedright}
\caption{(a) Comparison of the reliabilities of classifiers found using \ref{opt:p} and the robust formulation \ref{eq:rob_primal}, for the training data $\D_{\train}$ itself.
(b) The relative accuracy $\textrm{RA}(h^{*},f^{*})$ for the same dataset.
Looking at (a) and (b) together, (i) the relative accuracy of $h^{*}$, and (ii) the reliability $\widehat{\R}(h^{(N)},h^{*})$, become 1 much before the reliability $\widehat{\R}(f^{(N)},f^{*})$ reaches 1.}
\end{figure}

\subsubsection{Robust to Sampling Noise}
Based on the arguments presented, we expect the reliability $\R(h^{(N)},h^{*})$ to be higher than $\R(f^{(N)},f^{*})$ for the same value of $N$. 
Fig.\ \ref{fig:rel_compare} shows $\widehat{\R}(h^{(N)},h^{*})$ and $\widehat{\R}(f^{(N)},f^{*})$, as functions of $N$ for the Circles dataset:
it is clear that the robust classifier $h^{(N)}$ reproduces $h^{*}$ far more reliably than $f^{(N)}$ reproduces $f^{*}$, especially in the high noise regime (low $N$).
In particular, $\widehat{\R}(h^{(N)},h^{*})=1$ even with a single measurement shot $N=1$.

\subsubsection{Shot-Frugality}
Our original aim was to implement the true classifier $f^{*}$ with as few measurement shots per entry as possible, and not just to have a classifier robust to sampling noise.
\begin{definition}[{\bf Relative Accuracy}]
The relative accuracy of \ref{eq:rekc*} with respect to \ref{eq:ekc*} over any dataset is
\begin{equation}
\textrm{RA}\left(h^{*},f^{*}\right) = \frac{\textrm{Acc}\left(h^{*}\right)}{\textrm{Acc}\left(f^{*}\right)}.
\end{equation}
\end{definition}

\paragraph{Reliability-Accuracy trade-off.}
\cref{fig:rel_compare,fig:rel_acc} shows the reliability $\widehat{\R}(h^{(N)},h^{*})$ and the relative accuracy $\textrm{RA}(h^{*},f^{*})$, for the Circles dataset.
For small $N$, the classifier $h^{*}$ has low accuracy, but $h^{(N)}$ reproduces $h^{*}$ perfectly.
The low accuracy, for small $N$, can be understood from the expression for $J_{\rob}$ in Eq.(\ref{eq:J_rob}); the data-independent terms containing $\delta_1,\delta_2$ dominate over the data-dependent terms containing $\bK^*$, leading to fixed but inaccurate results.
With increasing $N$, increases in $\textrm{RA}(h^{*},f^{*})$ are accompanied by slight decreases in the reliability $\widehat{\R}(h^{(N)},h^{*})$, constituting a trade-off between the two quantities.
Finally, both these quantities reach their maximum value $1$.

\paragraph{Large savings in $N$.}
The smallest value of $N$ required to reliably agree with the results of \ref{eq:ekc*} is tabulated in Table \ref{tab:Reliability}, for both \ref{eq:skc*} and \ref{eq:rskc*}.
It is clearly seen that \ref{eq:rskc*} needs just a fraction of measurements (ranging from $\frac{1}{64}$ to $\frac{1}{4}$ for our datasets) to reliably reproduce \ref{eq:ekc*}.
Using \ref{eq:rob_primal} in place of the Naive SVM \ref{opt:p} leads to significant savings in $N$.

\ref{eq:rob_primal} is Second-Order Cone (SOC) representable (details in Appendix \ref{sec:socp}), and several efficient algorithms have been developed for such optimization problems \cite{Lobo1998socp}.
SOC programs have found applications in many engineering areas such as filter design, truss design and grasping force optimization in robotics (see\ Ref.\ \cite{Lobo1998socp}
for a survey of applications). SOC programs have also been very useful for many machine learning applications e.g.\ robust classification under uncertainty \cite{Lanckriet2003,Bhadra2010,Ben-Tal2011}.

%
\begin{table}
  \centering
  \renewcommand{\arraystretch}{1.2}
\begin{tabular}{|p{2cm}|p{1.2cm}|p{1.2cm}|p{1.2cm}|p{1.2cm}|}
\hline
\multirow{3}{*}{\thead{Dataset}} & \multicolumn{4}{c|}{\thead{$N$ needed for perfect reliability\\ with \ref{eq:ekc*}}} \\
\cline{2-5}
 & \multicolumn{2}{c|}{\textbf{GATES}} & \multicolumn{2}{c|}{\textbf{SWAP}} \\
\cline{2-5}
 & {\ref{eq:skc*}} & {\bf {\ref{eq:rskc*}}}& {\ref{eq:skc*}}& {\bf \ref{eq:rskc*}}\\
\hline
Circles & $2^{13}$ $\ddagger$ & $\mathbf{ 2^{8}}$ & $2^{16}$ $\ddagger$ & $\mathbf{ 2^{10}}$   \\
\hline
Havlicek \cite{Havlicek2019} & $2^7$ & $\mathbf{ 2^{4}}$ & $2^9$ & $\mathbf{ 2^{7}}$  \\
\hline
Two Moons & $2^{13}$ $\ddagger$ & $\mathbf{ 2^{6}}$ & $2^{15}$ $\ddagger$& $\mathbf{ 2^{8}}$  \\
\hline
Checkerboard \cite{Hubregtsen2021} & $2^{14}$ $\ddagger$ & $\mathbf{ 2^{8}}$ & $2^{16}$ $\ddagger$& $\mathbf{ 2^{10}}$  \\
\hline

\end{tabular}

\captionsetup{justification=raggedright}
\caption{The number of measurements $N$ at which perfect reliability\footnote{The values for \ref{eq:skc*} must coincide with $N_{\textrm{practical}}$ in Table\ \ref{tab:N_prac}, $N= 2^{n}$, $n = \lceil{\log_{2}(N_{\textrm{practical}})}\rceil$. The algorithm \ref{alg:N_prac} gives a more precise value for $N_{\practical}$ between $2^{n-1}$ and $2^{n}$.}
with respect to \ref{eq:ekc*} is achieved over the training dataset for classifiers \ref{eq:skc*} and \ref{eq:rskc*}.
The instantiations of the kernel function $K^{(N)}$ are the same for both.
The robust classifier reduces the number of circuit evaluations needed (per kernel entry) by a factor of $4$ to $64$.
$\ddagger$ denotes that $N$ exceeds the IBMQ default value $4000$.}
\label{tab:Reliability}
\end{table}

\subsection{Statistical estimate of the kernel matrix}
\label{sec:stat_est}

Until now, we worked under the assumption that the training kernel matrix was known exactly. In reality, however, it will
also have to be inferred from measurement statistics.
\begin{definition}[Estimated Kernel Matrix]
Over a training dataset of size $m$ the estimated kernel matrix $\widehat{\K}\in\RR^{m}\times\RR^{m}$
is a matrix, each of whose entries $\widehat{\K}_{ij} = K^{(T)}_{ij}$ are sampled once from
the distributions \ref{gates_dist} or \ref{swap_dist}, using $T$ measurements. $\widehat{\K}$ is then assumed to be fixed.
\end{definition}
Thus, the estimated kernel matrix is {\em fixed}
after a single estimation (using $T$ measurement shots) and is the same as a single instantiation of the stochastic kernel matrix entering Theorem \ref{thm:rob}. The latter is a random variable which we used to formulate the chance constraints.

\begin{remark}(Confidence Interval)
The true kernel matrix lies within the interval
$\K^{*}_{ij}\in [\widehat{\K}_{ij}-\Delta,\widehat{\K}_{ij} + \Delta]$  with probability of at least $1-2\delta$, where $\widehat{\K}_{ij} = K^{(T)}_{ij}$ denotes an estimated kernel matrix entry using $T$ measurement shots and
\begin{equation}
\label{eq:conf_int}
\Delta = \Delta(\delta,T) = \frac{c}{2\sqrt{T}}\kappa\left({\delta}\right)
\end{equation}
\end{remark}
\begin{proof}
Using \ref{fact:tail}, one can derive the double-sided bound
\begin{equation}\nonumber
\begin{aligned}
\Prob\left(|K^{*}_{ij}-K^{(T)}_{ij}|\geq\Delta\right)&\leq 2 \exp\bigg(-\frac{\Delta^2}{2\sigma^2}\bigg)\\
&\leq 2\delta
\end{aligned}
\end{equation}
and use the expression for $\sigma$ from Definition \ref{lemma:cf}.
Note that a subgaussian distribution has infinite support even though our quantum kernels
have support $[0,1]$.
\end{proof}

We derive the following Robust ``Estimate" Program from Theorem \ref{thm:rob} by using the estimated matrix $\widehat{\K}$ in place of $\K^{*}$.

\begin{thm}[Shot-frugal and Robust program using an Estimated kernel matrix]
\label{thm:rob_est}
Let $N,\delta_{1},\delta_{2}$ be fixed to the same values as in \ref{eq:rob_primal}.
Let $\widehat{\K}$ denote an Estimated Kernel Matrix using $T$ measurements.
The optimal \ref{eq:rob_primal} objective is upper bounded by
\begin{equation}
J^{*}_{\rob} \leq \widehat{J}^{*}_{\rob},
\end{equation}
with probability at least $1-\delta'_{1}-\delta'_{2}$.
Here
\begin{equation}\tag{\tt{ShofaR-Est}}
\label{eq:J_rob_l2_est}
\widehat{J}^{*}_{rob}= \min_{\beta,b,\xi,t} \widehat{J}_{\rob}(\beta,b,\xi,t;\widehat{\K}|T,\delta'_{1},\delta'_{2})
\end{equation}
is the optimal solution to the following optimization problem
with the objective
\begin{equation}
\widehat{J}_{\rob}(\xi,t,\beta,b;\widehat{\K}|T,\delta'_{1},\delta'_{2}):= C\sum_{i}^{m}\xi_{i}+t
\end{equation}
satisfying
$\xi_{i}\geq0 \ \ \forall i \in [m]$, and also
\begin{equation}
\label{eq:ssc1}
\begin{aligned}
\xi_{i}\geq 1 -  y_i\left(\sum_{j=1}^m \beta_{j} \widehat{\K}_{ij}+b\right) &+ \frac{\c}{2\sqrt{N}} \kappa\left(\frac{\delta_{1}}{m}\right) \norm{\beta}_{2}&\\
&+ \Delta\left(\frac{\delta'_{1}}{m},T\right)\norm{\beta}_{2}&
\end{aligned}
\end{equation}
and
\begin{equation}
\label{eq:ssc2}
t\geq\frac{1}{2} \beta^{\top}\left(\widehat{\bK} + \frac{\c}{2\sqrt{N}}\kappa(\delta_{2})  \mathbf{I}_{m\times m}+
\Delta(\delta'_{2},T)  \mathbf{I}_{m\times m}\right)\beta.
\end{equation}
\end{thm}

\paragraph{Ensuring Convexity.}
Note that above program may not be convex since $\widehat{\K}$ may have negative eigenvalues. Therefore our approach
is to make $\delta'_{2}$ small enough such that the matrix in the brackets of Eq.\ (\ref{eq:ssc2}) is always positive
semidefinite. This can be ensured for {\em any} value of $T$, however small [see Eq.\ (\ref{eq:conf_int})]. In our experiments, the parameteres $\delta_{1}=\delta_{2}=0.01$ in \ref{eq:rob_primal},
and we set $\delta'_{1} = 0.01$ in \ref{eq:J_rob_l2_est}.
The confidence interval entering the constraint (\ref{eq:ssc1}) is set to $ \Delta\left(\frac{\delta'_{1}}{m},T\right)=0.1$, and working backwards gives $T\sim 400$ for our training sizes.
 These are very reasonable values for $T$ (IBMQ default is 4000) and are a {\em one-time cost}. Note that the confidence interval depends weakly on the training size due to the logarithmic dependence. This parametrization of the program using the confidence interval makes it robust even in the presence of other sources of noise (see Appendix \ref{sec:depol}).

\paragraph{Assumption on Training Kernel Matrix in related work.}

Note that a positive semidefinite training kernel matrix is a requirement to solve the Naive SVM as well since it ensures a convex objective. The exact kernel matrix $\K^*$ is guaranteed to be positive semidefinite but a kernel matrix $\K^{(T)}$ over the training data constructed using $T$ samples per entry would generally not be positive semidefinite.
Previous studies \cite{Hubregtsen2021,Wang2021nisq} have circumvented this by either clipping the negative part of the spectrum of $\K^{(T)}$ or adding the largest negative eigenvalue to make it positive definite.
By contrast, the parameter $\delta'_{2}$ entering \ref{eq:J_rob_l2_est}
can always be made small enough to ensure a convex objective. It thus avoids the clipping of the negative spectrum in $\widehat{\K}$ and leads to a principled formulation.
\begin{figure*}
\centering
     \begin{subfigure}[hb]{0.48\textwidth}
         \centering
         \includegraphics[width=\textwidth,trim = 0cm 0cm 0cm 0cm]{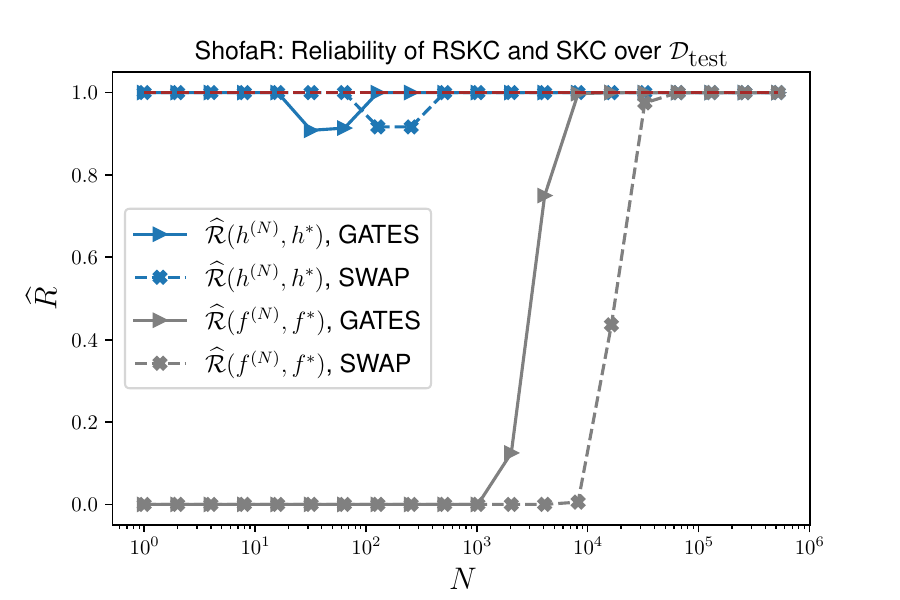}
         \caption{}
              \label{supp_fig:rel_shofar}

     \end{subfigure}
    \begin{subfigure}[hb]{0.48\textwidth}
    \centering
    \includegraphics[width=\textwidth,trim = 0cm 0cm 0cm 0cm]{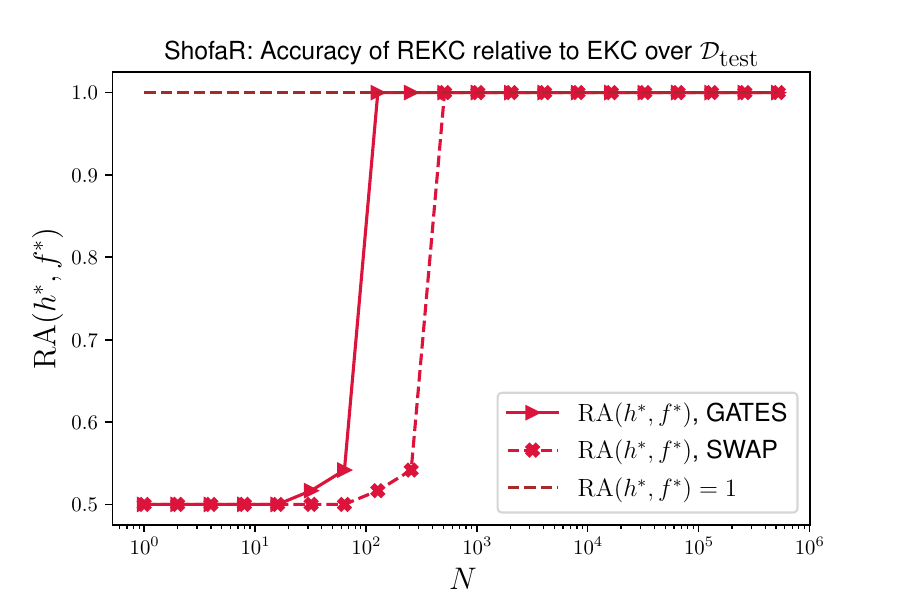}
    \caption{}
              \label{supp_fig:rel_acc_shofar}

     \end{subfigure}
     \begin{subfigure}[hb]{0.48\textwidth}
         \centering
         \includegraphics[width=\textwidth,trim = 0cm 0cm 0cm 0cm]{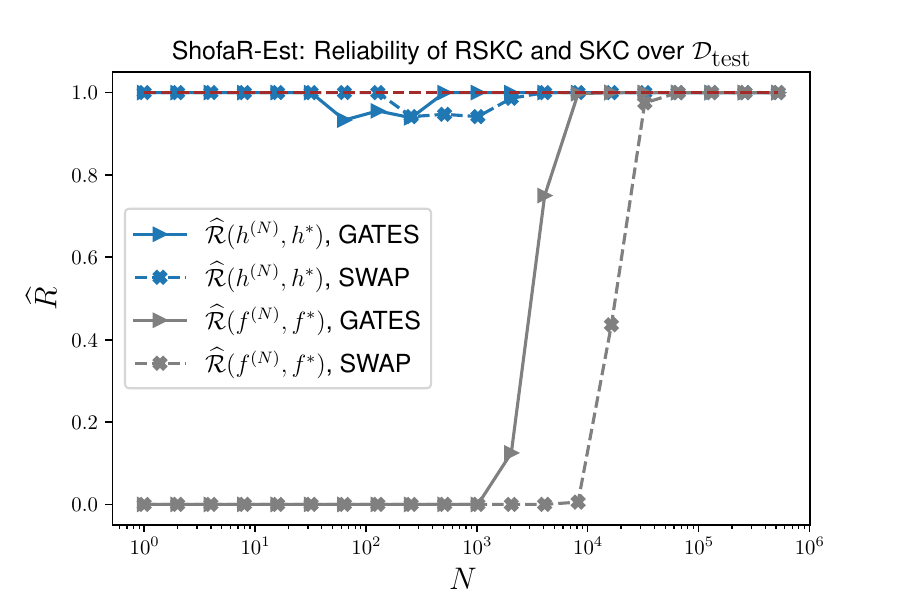}
         \caption{}
              \label{supp_fig:rel_shofar_est}

     \end{subfigure}
    \begin{subfigure}[hb]{0.48\textwidth}
    \centering
    \includegraphics[width=\textwidth,trim = 0cm 0cm 0cm 0cm]{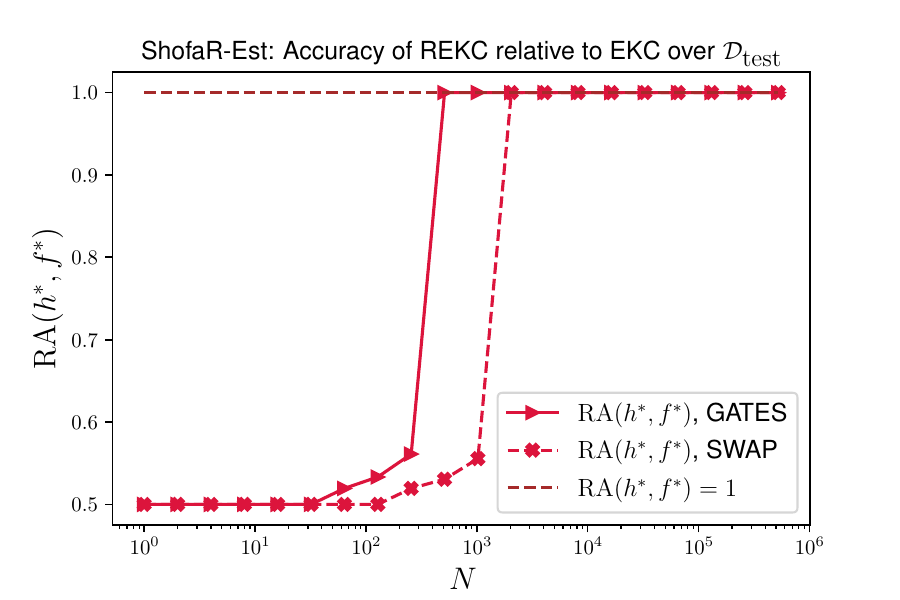}
    \caption{}
              \label{supp_fig:rel_acc_shofar_est}

     \end{subfigure}

     \captionsetup{justification=raggedright}
     \caption{Comparison of the \ref{eq:rob_primal} (top panel) and \ref{eq:J_rob_l2_est} (bottom panel) Programs over a test set $\D_{\test}$. The top panel compares the (a) reliability, and (b) relative accuracies, of the stochastic kernel classifiers arising from the \ref{eq:rob_primal} and \ref{opt:p} Programs.
     The bottom panel compares the (c) reliability, and (d) relative accuracies, of the stochastic kernel classifiers arising from the \ref{eq:J_rob_l2_est} and \ref{opt:p} Programs. It is seen that both \ref{eq:rob_primal} and \ref{eq:J_rob_l2_est} lead to significant savings in $N$
     over the nominal SVM obtained by solving \ref{opt:p}. The measurement requirements of \ref{eq:J_rob_l2_est} are marginally higher than \ref{eq:rob_primal} to reproduce \ref{eq:ekc*}. \ref{eq:J_rob_l2_est} uses an estimated kernel matrix
     but yet significantly outperforms the nominal \ref{opt:p}, even though the latter has access to the exact kernel matrix
     over the training data.}
\end{figure*}
\paragraph{Performance of the ``Estimate" Program.}
The Definitions \ref{eq:rekc*} and \ref{eq:rskc*} are naturally extended for the Robust Program \ref{eq:J_rob_l2_est}
by specifying $T,\delta'_{1},\delta'_{2}$; the values of $N,\delta_{1},\delta_{2}$ are assumed to be the same as in \ref{eq:rob_primal}. Over an independent test set, figs.\ \ref{supp_fig:rel_shofar_est} and \ref{supp_fig:rel_acc_shofar_est} respectively show the reliability of \ref{eq:rskc*} and \ref{eq:skc*}, and the relative accuracy of \ref{eq:rekc*}, for the \ref{eq:J_rob_l2_est} Program.
Figs.\ \ref{supp_fig:rel_shofar} and \ref{supp_fig:rel_acc_shofar} show the corresponding figures for the \ref{eq:rob_primal} Program. Both the Robust Programs as well as the nominal \ref{opt:p} SVM have the same stochastic kernel instantiations.
The Estimate program \ref{eq:J_rob_l2_est} leads to marginally larger values of $N$ to reliably reproduce \ref{eq:ekc*} but still shows large savings over the nominal \ref{opt:p} SVM. The performance over an independent test set for the different datasets are tabulated in \ref{tab:compare_robust}.

Importantly, our results (see Table\ \ref{tab:compare_robust}) show that \ref{eq:J_rob_l2_est} has a far lower requirement on $N$ than the nominal SVM (\ref{opt:p}) even though the latter has an unfair advantage: it is given access to the exact training kernel matrix $\K^{*}$ whereas the former only works with an estimate $\widehat{\K}$. We only tabulate $N$ from stochastic kernels computed using \ref{gates_dist} circuit. In our simulations in the presence of depolarizing noise, in Appendix \ref{sec:depol}, we work with the estimate $\widehat{\K}$ for the nominal classifiers as well. This results in larger resource requirements for the nominal SVM classifiers and worsens its accuracy.

\begin{table*}
  \centering
  \renewcommand{\arraystretch}{1.2}
\begin{tabular}{|p{2.2cm}|p{2.5cm}|p{2.5cm}|p{2.5cm}|p{2.5cm}|p{2.5cm}|}
\hline
\multirow{3}{*}{{\thead\ \ \ \ \ DATASET  }}
& \multicolumn{5}{c|}{\thead{$N,m_{\sv}$\\ For the stochastic kernel classifier
to perform same as \ref{eq:ekc*}}} \\
\cline{2-6}
 & {\thead{\ref{opt:p} \\ (Nominal SVM) }}& \thead{\ref{eq:rob_primal}}& \thead{\ref{eq:J_rob_l2_est}}& \thead{\ref{eq:J_rob_l1}} & \thead{\ref{eq:J_rob_l1_est}}\\
\hline
\thead{Circles\\ $m=40$}& \thead{$2^{14},8$}  & {\thead{\cellcolor{violet!90}${ 2^{7},40}$}} & \thead{\cellcolor{violet!75}$2^{9},40$} & \thead{\cellcolor{violet!50}$ 2^{12},6$} & \thead{\cellcolor{violet!50}$2^{13},9$}\\
\hline
\thead{Havlicek \\ $m=40$}  & \thead{$2^9,9$} & \thead{\cellcolor{violet!75}${ 2^{5},39}$} & \thead{\cellcolor{violet!90}$2^{4},40$ } & \thead{\cellcolor{violet!50}${ 2^{9},6}$}  & \thead{$2^9,9$}\\
\hline
\thead{Two Moons \\ $m=50$} & \thead{$2^{12},7$} & \thead{\cellcolor{violet!90}$ 2^{6},48$} & \thead{\cellcolor{violet!75}$2^{7},48$} & \thead{\cellcolor{violet!75}$ 2^{10},7$} &
\thead{\cellcolor{violet!25}$2^{11},10$}\\
\hline
\thead{Checkerboard \\
$m=100$}& \thead{$2^{16},20$} & \thead{\cellcolor{violet!75}$ 2^{9},94$} & \thead{\cellcolor{violet!50}$2^{11},94$} & \thead{$2^{15},22$} & \thead{$2^{17},15$}\\
\hline

\end{tabular}
\captionsetup{justification=raggedright}
  \caption{$N$ at which perfect reliability, with respect to $f^{*}$, is achieved over the test dataset
   have been listed along with the number of support vectors $m_{\sv}$
 which each program results in.
 The estimated kernel for the Programs \ref{eq:J_rob_l2_est} and \ref{eq:J_rob_l1_est} have different confidence intervals $\Delta({\frac{\delta'_{1}}{m},T}) = 0.1\ \textrm{and}\ 0.01$
 respectively. The instantiations
  of the kernel function $K^{(N)}$ are the same for all the $5$ stochastic kernel classifiers.
  The total number of circuit evaluations needed to reliably classify a given datapoint $N_{\textrm{tot}} = m_{\sv}N$.
  The Robust Programs generally result in large savings in $N_{\textrm{tot}}$, when compared against the nominal SVM,
  and have been colored in 4 shades representing savings of at least $25-50\%$, $50-75\%$, $75-90\%$ and $>90\%$. Darker shade
  represents greater savings. \ref{eq:J_rob_l2_est} does far better than the nominal \ref{opt:p} SVM even though the latter has access to the exact training kernel matrix.
  }
\label{tab:compare_robust}
\end{table*}



\subsection{Sparsity of support vectors}

So far we only looked at $N$, which is the number of measurements needed to estimate {\em one} kernel function.
The total number of circuit measurements needed for the classification of a single datapoint is $m_{\sv}N$, since $m_{\sv}$
kernel functions have to be evaluated.
The L2-norms over $\beta$ appearing in Theorems \ref{thm:rob} and \ref{thm:rob_est} lead to a number of support vectors
$m_{\sv}\sim \OO(m)$ of the order of the training set. Using the fact that $\norm{\beta}_{2}\leq\norm{\beta}_{1}$,
a relaxation of the L2-norm condition
into an L1-norm leads to greater sparsity of the set of support vectors $m_{\sv}\ll m$.
It is thus a possibility
that the resulting L1-norm Programs lead to a greater savings in the {\em total} number of circuit measurements.
We therefore state the following Corrolaries to Theorems \ref{thm:rob} and \ref{thm:rob_est} respectively.

\addtocounter{thm}{-1}
\begin{cor}(L1-norm Shot-frugal and Robust Program)
\label{thm:rob_l1}
For a stochastic training kernel matrix $\K^{(N)}\in\RR^{m}\times\RR^{m}$ inferred from $N$ measurements per entry,
the optimal \ref{opt:p} objective $J^{*}\big(\K^{(N)}\big) := \min_{\beta,b} J\big(\beta,b;\K^{(N)}\big)$ is upper bounded by
\begin{equation}
\label{eq:whp_rob}
J^{*}(\bK)\leq J^{*}_{\rob}\leq J^{*}_{\L1-\rob},
\end{equation}
with probability at least $1-\delta_{1}-\delta_{2}$, where
\begin{equation} \label{eq:J_rob_l1}
J^{*}_{{\L1}-\rob} = \min_{\beta,b,\xi,t } J_{{\L1}-\rob}\left(\xi,t,\beta, b; \bK^*| N, \delta_{1},\delta_{2} \right) \tag{\tt{L1-ShofaR}}
\end{equation}
is the optimal value of the Cone Program, whose objective is given by
\begin{equation}
J_{{\L1}-\rob}\left(\xi,t,\beta, b; \bK^*| N,\delta_{1},\delta_{2}\right) := C \sum_{i}\xi_{i}+t
\end{equation}
subject to
\begin{equation}\tag{L1-suff-chance-1}
\label{eq:sc1_l1}
\xi_{i}\geq 1 + \frac{\c}{2\sqrt{N}} \kappa{\bigg(\frac{\delta_{1}}{m}\bigg)}\norm{\beta}_{1} - y_i(\sum_{j=1}^m \beta_{j}\bK^{*}_{ij} + b)\ \ \ \forall i \in[m]
\end{equation}
\begin{equation}
\xi_{i}\geq0 \ \ \ \forall i \in[m]
\end{equation}
and
\begin{equation}\tag{suff-chance-2}
\label{eq:sc2_l1}
t\geq\frac{1}{2} \beta^{\top}\left(\bK^* + \frac{\c}{2\sqrt{N}}\kappa(\delta_{2})  \mathbf{I}_{m\times m}\right)\beta
\end{equation}
\end{cor}
\begin{proof}
(\ref{eq:sc1_l1}) is a sufficient condition to ensure (\ref{eq:sc1}) since
$\norm{\beta}_{2}\leq\norm{\beta}_{1}$. This also implies that $J_{\rob}\leq J_{L1-\rob}\ \ \forall\beta$.
\end{proof}
In the case where the kernel matrix is estimated using $T$ measurements per entry, we have the following result.
\addtocounter{thm}{1}
\addtocounter{cor}{-1}
\begin{cor}[L1-norm Shot-Frugal and Robust Program using an Estimated kernel matrix]
\label{thm:rob_est_l1}
Let $N,\delta_{1},\delta_{2}$ be fixed to the same values as in \ref{eq:rob_primal}.
Let $\widehat{\K}$ denote an Estimated Kernel Matrix using $T$ measurements.
The optimal \ref{eq:rob_primal} objective is upper bounded by
\begin{equation}
J^{*}_{{\L1}-\rob} \leq \widehat{J}^{*}_{{\L1}-\rob},
\end{equation}
with probability at least $1-\delta'_{1}-\delta'_{2}$.
Here
\begin{equation}\tag{\tt{L1-ShofaR-Est}}
\label{eq:J_rob_l1_est}
\widehat{J}^{*}_{{\L1}-rob}= \min_{\beta,b,\xi,t} \widehat{J}_{{\L1}-\rob}(\beta,b,\xi,t;\widehat{\K}|T,\delta'_{1},\delta'_{2})
\end{equation}
is the optimal solution to the following optimization problem
with the objective
\begin{equation}
\widehat{J}_{{\L1}-\rob}(\xi,t,\beta,b;\widehat{\K}|T,\delta'_{1},\delta'_{2}):= C\sum_{i}^{m}\xi_{i}+t
\end{equation}
satisfying
$\xi_{i}\geq0 \ \ \forall i \in [m]$, and also
\begin{equation}
\label{eq:ssc1_l1}
\begin{aligned}
\xi_{i}\geq 1 -  y_i\left(\sum_{j=1}^m \beta_{j} \widehat{\K}_{ij}+b\right) &+ \frac{\c}{2\sqrt{N}} \kappa\left(\frac{\delta_{1}}{m}\right) \norm{\beta}_{1}&\\
&+ \Delta\left(\frac{\delta'_{1}}{m},T\right)\norm{\beta}_{1}&
\end{aligned}
\end{equation}
and
\begin{equation}
\label{eq:ssc2_l1}
t\geq\frac{1}{2} \beta^{\top}\left(\widehat{\bK} + \frac{\c}{2\sqrt{N}}\kappa(\delta_{2})  \mathbf{I}_{m\times m}+
\Delta(\delta'_{2},T)  \mathbf{I}_{m\times m}\right)\beta.
\end{equation}
\end{cor}

The total number of measurements needed to reliably reproduce \ref{eq:ekc*} for the SKCs arising from the nominal \ref{opt:p} SVM and the $4$ ShofaR Programs,
namely, \ref{eq:rob_primal}, \ref{eq:J_rob_l2_est}, \ref{eq:J_rob_l1} and \ref{eq:J_rob_l1_est}, are given in Table\ \ref{tab:compare_robust}. We only list the values for the \ref{gates_dist} circuit since the corresponding values for
the \ref{swap_dist} are roughly 4 times larger (see Table \ref{tab:test_reliability}).

\subsection{Benchmarks in literature}
\label{sec:lit_compare}

Clearly stated benchmarks of $N$ are hard to come by in the existing literature.
The reason for this is two-fold: (a) The problem of unreliable classification has not been given due importance and
(b) there are other device noises which adds an overhead to measurement requirements, making it hard to state clear benchmarks.
While some studies have considered both device and sampling noise \cite{Havlicek2019,Hubregtsen2021,Wang2021nisq} in the training kernel matrix,
 the measurement requirements during the classification phase has not received much attention.
%
We discuss the bounds obtained by requiring precise kernel evaluations and their relation to refs.\ \cite{Havlicek2019,Hubregtsen2021,Wang2021nisq} in Appendix \ref{sec:precise}. In Appendix \ref{sec:depol}, we include simulations in the presence of the depolarizing noise. This demonstrates that the Robust Programs we derived outperform the nominal SVM classifiers in the presence of other sources of noise as well.

\section{Summary and Outlook}

We have studied the role of $N$, the number of measurements used to evaluate a quantum kernel, for a classification task.
Our considerations have focused on providing the same predictions as the ideal quantum kernel classifier ($N\to\infty$), {\em not} to closely approximate the kernel function itself.
We noted that the classification accuracy is a poor performance metric in the presence of noise, and defined an empirical reliability that meaningfully captures the effects of noise.
The circuit used to evaluate the quantum kernel plays an important role in the analysis, and we have shown that the GATES test is preferable to the SWAP test for any $N$ (Definition \ref{lemma:cf}).
We introduced a generic uncertainty model (\ref{eq:um}) that can handle any source of noise.
Our results in the article have considered only the fundamental sampling noise;
the modification to the uncertainty model in presence of other noise sources is given in Appendix \ref{sec:datasets}. We show simulations in the presence of the depolarizing noise in Appendix \ref{sec:depol}.

Using a subgaussian bound, we showed that $N_{\theory} \sim m_{\sv}\log M/\gamma^{2}$ measurements are sufficient for reliable reproduction of the ideal quantum classifier over any dataset of size $M$.
If all pairs of kernel entries of this dataset and the set of support vectors are evaluated to a precision of $\varepsilon$ (in Frobenius norm), then the bound scales as $N \sim m_{\sv}M/\varepsilon^{2}$ (see Appendix \ref{sec:precise} for the derivation and comparison to related work).
We noted that the parameter $\gamma$, representing the margin of classification of an ideal quantum classifier,
plays a role analogous to the precision $\varepsilon$ with which a kernel entry is evaluated. Crucially, $\gamma$ is not vanishingly small. We then defined a tighter empirical bound $N_\practical$ and provided an algorithm for its computation in Appendix \ref{alg:N_prac}.

In Section\ \ref{sec:expt_2}, we developed the chance constraint programming for the SVM primal formulation and gave a Shot-Frugal and Robust (\ref{eq:rob_primal}) Program in Theorem \ref{thm:rob}. We showed that the resulting classifier requires far fewer circuit evaluations than even the empirical bound $N_{\practical}$ to reliably reproduce the ideal quantum classifier. In Theorem \ref{thm:rob_est}, we derive an additional program \ref{eq:J_rob_l2_est} which takes into account the fact that we do not have the exact training kernel matrix. We showed that this still leads to large savings over $N_{\practical}$, even when compared with the nominal \ref{eq:rob_primal} SVM which has access to the exact training kernel matrix. 
In Corollaries \ref{thm:rob_l1} and \ref{thm:rob_est_l1}, we consider an L1-norm relaxation of the \ref{eq:rob_primal} and \ref{eq:J_rob_l2_est} Programs respectively in order to reduce the set of support vectors. Table\ \ref{tab:compare_robust} compares the measurement requirements of all the ShofaR Programs against the nominal SVM classifier.
%

Finally, we point out some open questions related to our results that are worth exploring,
especially in light of the limitations of near-term quantum hardware.
\paragraph{Problem dependence.} To find the tightest theoretical bound $N_{\theory}$ for an accurate stochastic classifier \ref{eq:skc*}, we need to maximize $\gamma$ while maintaining low $\gamma$-margin errors of the exact classifier \ref{eq:ekc*}, i.e.\ stay just left of the vertical green line in Fig.\ \ref{fig:N_prac}.
Such a value of $\gamma$ would be problem (dataset + kernel choice) dependent, but can we develop a criterion to identify it?
Additionally, finding the lowest $N$ for a robust stochastic classifier \ref{eq:rskc*} while ensuring high relative accuracy, we need to estimate where $\textrm{RA}(h^{*},f^{*})$ reaches $1$ in Fig.\ \ref{fig:rel_acc}.
This would also be problem dependent, but can we find a method to identify it?
\paragraph{Possible extensions.} A kernel approach to understanding training evolution during gradient descent for deep neural networks was put forth recently \cite{jacot2018neural}. This method has been extended into the quantum domain \cite{shirai2022kernel,liu2022kernel} to understand the training of quantum neural networks during gradient descent.
Indeed, it is pointed out in Ref.\ \cite{liu2022kernel} that it would be worth while to explore the robustness of these kernel methods against noise. On a related note, a robust version \cite{chen2021rshadow} of the classical shadow estimation scheme \cite{aaronson2018shadow,huang2020shadow} has been put forth to efficiently learn properties of a many-body quantum system.
It would be interesting to explore the use of chance constraint programming in these topics.


\section*{Acknowledgments}
Our simulations with quantum kernels are performed using the {\tt Pennylane} python package developed by Xanadu \cite{Pennylane}.
We thank Kanika Gandhi for assistance in simulating the depolarizing noise and the authors of Ref.\ \cite{Hubregtsen2021} for making this code available.

AS was supported by the National Postdoctoral Fellowship, Science and Engineering Research Board (SERB), Government of India.
The work of AJ and AP was supported in part by the Centre for Excellence in Quantum Technology funded by the Ministry of Electronics and Information Technology, Government of India.

\section*{Code Availability}
The code for the numerical simulations can be found at \url{https://github.com/abhayshastry/ReliableQuantumKernelClassifiers}.
Correspondence regarding the manuscript and code should be sent to AS.

\bibliography{refs}
\begin{appendix}
\section{Practical bound computation for a specified margin}
\label{alg:N_prac}
We follow the algorithm below to compute the practical bound on $N$, from the Definition \ref{def:N_prac}. Fig.\ \ref{fig:N_prac} shows that this practical bound on $N$ is always below the theoretical bound given by Theorem \ref{thm:emp_risk}.

\begin{algorithm}[H]
\caption{Computing $N_{\mathrm{practical}}$}

\begin{algorithmic}
\Require{$\mathcal{D}$, $K^*$, $\gamma$}
\State{/* Given: $\beta^{*}$, $b^{*}$ over training data $\mathcal{D}_{\train}$ */}
\State{Initialize $N_{\mathrm{start}}$, $N_{\mathrm{step}}$, $\delta_{\target}=0.01$, $N_{\trials}=200$}
\For{$i \gets 1$ ~to~ $|\mathcal{D}|$}
    \For{$j \gets 1$ to $|\mathcal{D}_{\train}|$}
        \State{$\bK[i,j] \gets K^*(\bx_i, \bx_j)$}
    \EndFor
\EndFor
\qquad
\State{/* Above $K^{*}_{ij}$ needed for computation of $f^{*}(\bx_{i})$ */}

\State{$\epsilon_{\gamma} \gets \sum_{i \in [|\mathcal{D}|] } \mathbb{1}_{y_i f^{*}(\bx_i) < \gamma}$ }
\State{$N \gets N_{\mathrm{start}}$}
\qquad
\While{ $\delta_{\emp} \geq \delta_{\target}$}
\State{$N \gets N + N_{\mathrm{step}}$}
\For{$t \gets 1$ to $N_{\trials}$}
\State{$ c \gets 0 $}
    \For{$i \gets 1$ to $|\mathcal{D}|$}
        \For{$j \gets 1$ to $|\mathcal{D}_{\train}|$}
        \State{$\bK[i,j] \gets K^{(N)}\{t\}\left(\bx_i, \bx_j\right)$}
        \EndFor
    \EndFor
\State{/* $K^{(N)}\{t\}(\bx_{i},\bx_{j})$ for computing $f^{(N)}\{t\}(\bx_{i})$ */}
    \State{$\epsilon_0 \gets \sum_{i \in [|\mathcal{D}|] } \mathbb{1}_{y_i {f}^{(N)}\{t\}(\bx_{i}) < 0}$ }
    \If{$\epsilon_0 > \epsilon_\gamma$}
    \State{ $c \gets c + 1$}
    \EndIf
\EndFor
    \State{$\delta_{\emp} \gets {c}/{N_{\trials}}$}
\EndWhile
\Return{$N$}
\end{algorithmic}

\end{algorithm}

\section{Datasets and Kernel Embedding choices}
\label{sec:datasets}
\paragraph{\bf Datasets.} Our simulations with quantum kernels are performed using the {\tt PennyLane} python package developed by Xanadu \cite{Pennylane}. We use three datasets to illustrate our results: {\tt make_circles} (Circles) and {\tt make_moons} (Moons) dataset from {\tt sklearn} and the Havlicek generated dataset. The latter is generated by using {\tt IQPEmbedding}
and following the procedure outlined in ref.\ \cite{Havlicek2019}. We then use the same embedding kernel to classify the generated dataset. Our choice of kernel functions
for the other two datasets (Circles and Moons) are based on a trial and error process where we tried different embedding kernels and chose the one that returned
the highest test accuracies. Our choice of the kernel function for the dataset Circles is the simple {\tt AngleEmbedding} circuit shown in fig.\ (\ref{fig:angle_embed}) in the main
article. The Moons dataset was found to give very high test accuracies by using a X-rotation ({\tt AngleEmbedding}) after the application of the {\tt IQPEmbedding}.  All of these kernels are defined over 2-qubits (all datasets here are 2-dimensional) with no repeating layers and no tunable circuit parameters. We have chosen a smaller number of training samples, keeping
in view the restrictions placed by near-term devices \cite{Wang2021nisq}.

Table\ \ref{tab:datasets} shows training and test sizes. For the Havlicek dataset we use $m=40$ (20 labels per class as used in ref.\ \cite{Havlicek2019}), and a equal test size $M=40$, exactly as used in their study. All the training and test datasets used here are balanced. These kernel choices were used in the main article and all of them lead to nearly perfect test accuracies for the exact kernel classifier \ref{eq:ekc*}.

\paragraph{\bf Quantum Kernel embeddings.}

It is easy to see that the kernel function is equivalent to an inner
product in the {\em feature} Hilbert space $\mathcal{H}\otimes \mathcal{H}^*$ since,
for every pair $\ket{\phi(\bx_i)}, \ket{\phi(\bx_j)}\in \mathcal{H}$, we have
$\ket{\phi(\bx_i)}\otimes \ket{\phi^{*}(\bx_i)}, \ket{\phi(\bx_j)}\otimes \ket{\phi^{*}(\bx_j)} \in\mathcal{H}\otimes \mathcal{H} $, whose inner product is the quantum
kernel $K(\bx_{i},\bx_{j}) = |\bra{\phi(\bx_{i})}\phi(\bx_{j})\rangle|^{2}$. One can
express this in terms of an outer product $\ket{\phi(\bx_j)}\otimes \bra{\phi(\bx_j)} \in\mathcal{H}\otimes \mathcal{H^{*}} $
where the embedding vectors are density matrices and the kernel can be described via the
Hilbert-Schmidt inner product as given in (\ref{eq:QEK}).

\begin{figure}[hb]
\centering
\includegraphics[width=0.5\textwidth, trim = 2cm 21cm 9cm 5cm]{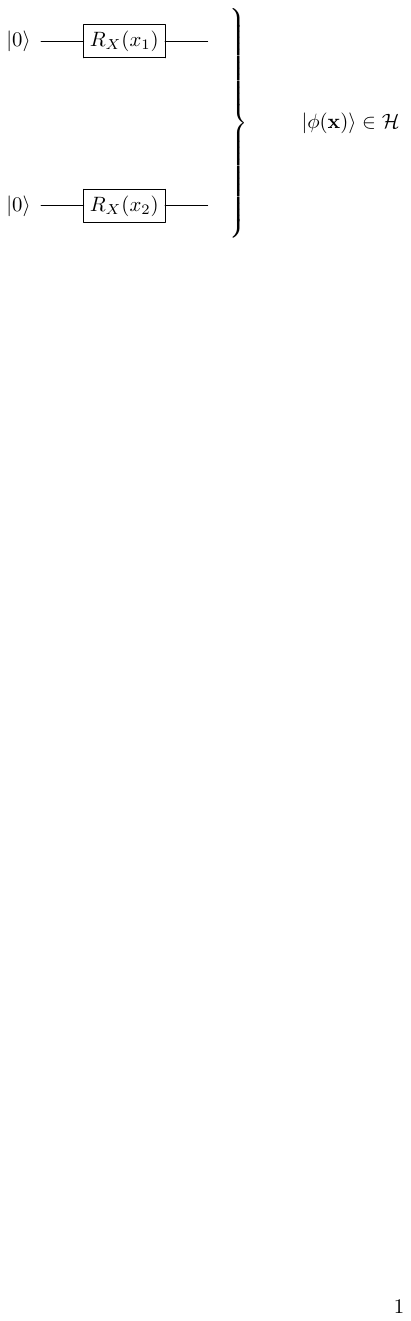}
\captionsetup{justification=raggedright}
\caption{{\bf Example: Angle Embedding.} A simple quantum embedding circuit which is embedding the input vector $\bx = (x_{1},x_{2}) \in\RR^{2}$
    into a quantum state $\ket{\phi(\bx)}$ in the 2-qubit Hilbert space $\mathcal{H}$,
    by applying a rotation along the X-axis by angle $x_{i}$ over the corresponding qubit starting
    in the state $\ket{0}$.
    The transition probability between two such quantum states $K(\bx,\bx') = |\langle{\phi(\bx')}|\phi(\bx)\rangle|^2$ is a valid kernel function and is known as a Quantum Embedding Kernel (QEK).
}
\label{fig:angle_embed}
\end{figure}

Note that the {\tt AngleEmbedding} shown in fig.\ \ref{fig:angle_embed} is a
simple quantum embedding circuit which can in fact be simulated efficiently on a classical machine \cite{Goldberg2017, Havlicek2019}. However, our study does not directly
concern the debate regarding quantum advantage. Rather, our focus is on the number of measurements $N$ needed to reliably implement a quantum kernel classifier.

All the figures presented in the main article used the Circles dataset with the {\tt AngleEmbedding} kernel.

\paragraph{\bf Change of kernel embedding}. Our results derived in the main article over the bounds on $N$ as well as the improved performance of the robust classifiers hold regardless of the choice of kernel. For a poorly chosen kernel, the accuracy of \ref{eq:ekc*} will be low and there will not exist a $\gamma^*$ satisfying Corollary \ref{cor:max_gamma}. At best, Theorem \ref{thm:emp_risk} guarantees that \ref{eq:skc*} will have an accuracy which is not worse than that of \ref{eq:ekc*}.

\begin{table*}
\def\arraystretch{1.5}
\begin{tabular}{|c|c|c|c|c|}
\hline
 \thead{Dataset }& \thead{Kernel Embedding}
 & \thead{Training size\\ $m$} & \thead{Test size \\ $M$ } & Test Accuracy\\
 \hline
  Circles (\tt make_circles) & {\tt AngleEmbedding}
 & 40 & 360 & 100\% \\ \hline
 Havlicek \cite{Havlicek2019} & {\tt IQPEmbedding}
 & 40 & 40 & 100\%\\
 \hline
 Moons (\tt make_moons) & \thead{{\tt IQPEmbedding} + \\ {\tt AngleEmbedding} }& 50 & 350 & 100\% \\
 \hline
  Checkerboard \cite{Hubregtsen2021} & {{\tt IQPEmbedding} }& 100 & 300 & 100\% \\
 \hline
\end{tabular}
\captionsetup{justification=raggedright}
\caption{Summary of datasets and quantum kernel embedding choices. The Moons dataset is encoded using the {\tt IQPEmbedding}
followed by a X-rotation by an angle given by the corresponding values of the feature vector. All embedding circuits are defined
over 2 qubits with no layer repetitions and no tunable parameters.
These choices lead to very high test accuracy for the true kernel classifier \ref{eq:ekc*}.}
\label{tab:datasets}
\end{table*}

\begin{table}
  \centering
  \renewcommand{\arraystretch}{1.2}
\begin{tabular}{|p{2cm}|p{1.2cm}|p{1.2cm}|p{1.2cm}|p{1.2cm}|}
\hline
\multirow{3}{*}{\thead{Dataset}} & \multicolumn{4}{c|}{\thead{$N$ needed for perfect reliability\\ with \ref{eq:ekc*}}} \\
\cline{2-5}
 & \multicolumn{2}{c|}{\textbf{GATES}} & \multicolumn{2}{c|}{\textbf{SWAP}} \\
\cline{2-5}
 & {\ref{eq:skc*}} & {\bf {\ref{eq:rskc*}}}& {\ref{eq:skc*}}& {\bf \ref{eq:rskc*}}\\
\hline
Circles & $2^{14}$ $\ddagger$ & $\mathbf{ 2^{8}}$ & $2^{16}$ $\ddagger$ & $\mathbf{ 2^{10}}$  \\
\hline
Havlicek \cite{Havlicek2019} & $2^9$ & $\mathbf{ 2^{5}}$ & $2^{11}$  & $\mathbf{ 2^{8}}$  \\
\hline
Two Moons & $2^{12}$ $\ddagger$ & $\mathbf{ 2^{6}}$ & $2^{13}$ $\ddagger$& $\mathbf{ 2^{9}}$  \\
\hline
Checkerboard \cite{Hubregtsen2021} & $2^{14}$ $\ddagger$ & $\mathbf{ 2^{9}}$ & $2^{16}$ $\ddagger$& $\mathbf{ 2^{12}}$ $\ddagger$ \\
\hline

\end{tabular}
\captionsetup{justification=raggedright}
  \caption{$N$ at which perfect reliability, with respect to $f^{*}$, is achieved over the test dataset for classifiers \ref{eq:skc*} and \ref{eq:rskc*}. Training test is the same
  as considered in the main article. The instantiations
  of the kernel function $K^{(N)}$ are the same for both the stochastic kernel classifiers.
  The \ref{eq:rob_primal} Program reduces the number of circuit evaluations needed (per kernel evaluation) by a factor of $8$ to $64$ in the datasets considered here. $\ddagger$ denotes $N$ exceeding IBM default value of $4000$.
  }
\label{tab:test_reliability}
\end{table}

\paragraph{\bf Other sources of noise.}
Recently, kernel estimates have been introduced that correct for depolarizing error as well \cite{Hubregtsen2021}. These estimates are not unbiased, i.e.\ $\EE[K^{(N)}\left(\bx_{i},\bx_{j}\right)] \neq K^*(\bx_{i},\bx_{j})$, and requires separate treatment. In such a case, we may write the final kernel function $K_{ij}$
as
\begin{equation}
K_{ij} = \widehat{K}_{ij} + \Delta{K}^{(N)}_{ij} + \Delta{K}^{(\textrm{other})}_{ij},
\end{equation}
where $\EE[{K_{ij}}] = \widehat{K}_{ij} \neq K^{*}_{ij}$ which must be estimated using
a bias-correcting procedure (such as the one for depolarizing noise in ref.\ \cite{Hubregtsen2021}). All other noise contributions have been subsumed into $\Delta{K}^{(\textrm{other})}_{ij}$ whose expectation value is taken as zero without loss of generality. Assuming they are independent sources of noise, we would have
\begin{equation}
\sigma_{0}^{2} = \sigma^{2}(K_{ij}) = \frac{\c^{2}}{4N} + \sigma^{2}(\textrm{other}),
\end{equation}
and the methods we develop in the article can be applied with the above $\sigma_{0}$ in the uncertainty model \ref{eq:um}. However, if the noise biases the kernel value, we would need to compare all our results to $\widehat{K}$ instead of $K^{*}$. All our results in the main article were in the presence of the sampling noise alone to clarify its contribution.
In Appendix \ref{sec:depol}, we consider the device noise.

\paragraph{\bf Performance over independent test set.}
Table \ref{tab:test_reliability} shows the $N$ needed per kernel entry. We can see large savings due to the \ref{eq:rob_primal}
Program.

\section{Subgaussian  property of Bernoulli random variables}
\label{sec:subgaussian}
A random variable $\xi$ (centered, $\EE {\xi}=0$) is called subgaussian if there exists $a\in [0,\infty)$ such that
\begin{equation}
\EE {e^{\lambda\xi}}\ \leq e^{\frac{a^{2}\lambda^{2}}{2}},\ \forall \lambda\in\mathbb{R}
\end{equation}
and the smallest such $a$
\begin{equation}
\sigma(\xi) = \textrm{inf}\ \big\{ a\geq0 \big| \EE {\exp{ \{ \lambda\xi \}}} \leq \exp{\big\{{a^{2}\lambda^{2}}/{2}\big\}},\ \forall \lambda\in\mathbb{R} \big\}
\end{equation}
is called the {\em subgaussian norm} of $\xi$ (smallest such $a^{2} = \sigma^{2}(\xi)$ has been called the optimal variance proxy).
We note two straight-forward properties of the subgaussian norm \cite{Buldygin2013}: (1) \emph{Scaling}, i.e. $\sigma^{2}(\lambda X) = \lambda^{2}\sigma^{2}(X)$ for any $\lambda\in\mathbb{R}$,
and (2) \emph{Sub-additivity}, i.e. for any independent $X_{k},\ k = \{1,2, ..., N\}$, $\sigma^{2}(\sum_{k=1}^{N}(X_k)) \leq \sum_{k=1}^{N} \sigma^{2}(X_{k})$.
Equality holds for independent and identical random variables $\xi_{k},\ k = \{1,2, ..., N\}$, $\sigma^{2}(\sum_{k=1}^{N}\xi_{k}) = N \sigma^{2}(\xi_{1})$.

For both GATES and SWAP, the mean of the estimated kernel matches the true kernel $\EE[K^(N)(\bx_i, \bx_j)] = K^{*}(\bx_i, \bx_j)$.  The variances of these two estimates are given by
\begin{equation}
 \Var(K^{(N)}(\bx_i, \bx_j)) = \begin{cases} \frac{K^*_{ij} \big(1 - K^*_{ij}\big)}{N}\hspace{-0.1in}&\text{(GATES)} \\
 \frac{ \big(1 - (K^*_{ij})^2\big)}{N} ~~~ &\text{(SWAP)}.
 \end{cases}
 \label{var_gates_swap}
\end{equation}
In Eq.\ (\ref{var_gates_swap}) since $0\leq K^{*}_{ij}\leq1$, we have $K^*_{ij} (1 - K^*_{ij})\leq1 - (K^*_{ij})^2$, implying a lower variance for the GATES circuit for
the same $N$:
\begin{equation}
\label{ineq_gates_swap}
\Var\left(K^{(N)}(\bx_i, \bx_j)\Big{|}_{\textrm{\tiny SWAP}}\right)\geq\Var\left(K^{(N)}(\bx_i, \bx_j)\Big{|}_{\textrm{\tiny GATES}}\right)
\end{equation}

Thus the SWAP circuit is affected more by the sampling noise.
The maximum value of $K^*_{ij} (1 - K^*_{ij})$ is $0.25$ whereas $1 - (K^*_{ij})^2$ can take a maximum value of $1$.
This bounds the variance of $K^{(N)}_{ij}$ in Eq.\ (\ref{var_gates_swap}) (equivalently the variance
of $\Delta{K}^{(N)}_{ij}$).

We note that the distributions given by (\ref{gates_dist}) and (\ref{swap_dist}) are subgaussian.
Furthermore, we note that $\Delta{K}^{(N)}_{ij}$ in (\ref{DeltaK}) can be expressed as the sum of $N$ i.i.d.\ centered Bernoulli random variables $X_{k}(p)\sim B^{(0)}(1,p) = p - B(1,p),\ k=\{1,2, ..., N\}$ for both the
GATES and the SWAP tests
\begin{equation}
N\Delta{K}^{(N)}_{ij} = \begin{cases}
 \sum_{k=1}^{N}X_{k}(p),\ p = K^{*}_{ij} & (\textbf{GATES})\\
2\sum_{k=1}^{N}X_{k}(p),\ p = \frac{1+ K^{*}_{ij}}{2} & (\textbf{SWAP}).
\end{cases}
\label{DeltaK_ind}
\end{equation}
Stated another way, the prefactor of 2 above for the SWAP test would have the effect of increasing its variance and thus
makes it less reliable for the same number of measurements $N$.
Bernoulli random variables are subgaussian \cite{Buldygin2013,Arbel2019} and their properties are covered in section \ref{sec:subgaussian}. We use
these properties to derive tail bounds over the kernel distributions.

We omit the random variable $\xi\sim B^{(0)}(1,p)$ while denoting the subgaussian norm for a centered Bernoulli distribution
and simply
write $\sigma(p)$.
The optimal variance proxy for a centered Bernoulli distribution was derived in \cite{Buldygin2013} and is given by ($q = 1 - p$)
\begin{equation}
\label{eq:sg_norm_ber}
\sigma^{2}(p) = \begin{cases}
0 & \ \ p = \{0,1\}\\
0.25 & \ \ p = 0.5 \\
\frac{p-q}{2(\ln(p)-\ln(q))} & \ \ \textrm{otherwise}.
\end{cases}
\end{equation}
Its properties were derived in lemma 2.1 of \cite{Buldygin2013} from which we note that
$\sigma^{2}(p)$ is monotonically increasing for $p\in(0,1/2)$ and monotonically decreasing for $p\in(1/2,1)$. It is also symmetrical $\sigma^{2}(p) = \sigma^{2}(q)$ (where $q = 1-p$).
It takes on its maximum value of 1/4 when $p = q =  1/2$. Using the maximum value of the subgaussian norm for Bernoulli random variables above and
eq.\ (\ref{DeltaK_ind}) for the two different circuits gives us the bound used in (\ref{eq:cf}).

A recent study \cite{Nghiem2021} on the IBM quantum hardware has shown that SWAP test is a poor choice compared to the GATES test for estimating QEKs, owing to errors in the controlled SWAP implementation. While this may be a factor, it appears that the authors may have overlooked that a straightforward inequality (\ref{ineq_gates_swap}) underlies their observation.

\section{More on chance constraint programming}

\subsection{Chance constraint: Union bound}
\label{sec:union_bound}

When multiple chance constraints are present, we have invoked the well-known union bound to derive our results \cite{Boole2009}.
Suppose the chance constraints to be satisfied are some inequalities $A_{i},\ i\in[m]$, we denote its conjugate as $\bar{A}_{i}$ which is the violation of inequality $A_{i}$.
We would like all the constraints $A_{i}$, i.e.\ $\cap_{i}^{m}A_{i}$, to be satisfied and derive a sufficient condition for it. Since $\overline{\cap_{i}^{m}A_{i}}=\cup_{i}\bar{A}_{i}$,
we consider the probability of the union of all constraint violations
\begin{equation}
\begin{aligned}
\Prob\big(\bar{A}_{1}\cup\bar{A}_{2}\cup ...\bar{A}_{m}\big) &\leq\sum_{i=1}^{m}\Prob\big(\bar{A}_{i}\big),\\
&\leq\delta,
\end{aligned}
\end{equation}
for some small parameter $0<\delta\ll1$. This provides the sufficient condition
\begin{equation}
\begin{aligned}
\label{eq:ub}
\Prob(\bar{A}_{i})&\leq\frac{\delta}{m},\\
\Prob(A_{i})&\geq 1-\frac{\delta}{m},\ \ \forall i\in[m]
\end{aligned}
\end{equation}

\subsection{Sufficient condition in the Chance Contraint Program}
\label{sec:chance_primal}

We give a slightly more detailed explanation of the chance constraint approach in this Section, at the expense
of repetition from Section \ref{sec:chance}.
The kernel matrix $\K=\K^{(N)}= \K^{*} - \Delta\K^{(N)}$ is stochastic and entries of $\Delta{\K}^{(N)}$ satisfy
the uncertainty model (\ref{eq:um}) with the bound given by Eq.\ (\ref{eq:cf}).

The hinge loss given by
\begin{equation}
L(\beta,b;\bK) = \sum_{i=1}^{m} \max\bigg(1-y_{i}(\sum_{j}\bK_{ij}\beta_{j} + b), 0\bigg)
\end{equation}
is also subject to the stochasticity of the kernel matrix $\bK$. Whether the term is zero or positive depends on the particular instantiation of the kernel matrix.
This will have to be handled by a chance constraint approach. The same objective in the
SVM primal problem (\ref{eq:obj}) can be written as
\begin{equation}\tag{PRIMAL-2}
\label{eq:obj2}
\min_{t,\xi,\beta,b} J(t,\xi,\beta,b;\K):= C\sum_{i=1}^{m}\xi_{i} + t,
\end{equation}
with the hinge-loss term as the two constraints
\begin{equation}
\label{hl_con}
\IE_{i}:\ \xi_{i}\geq 1 - y_{i}\big(\sum_{j=1}^{m}\bK_{ij}\beta_{j} + b\big),
\end{equation}
and
\begin{equation}
\xi_{i}\geq 0,\ i\in[m].
\end{equation}
The quadratic term is translated to the constraint
\begin{equation}
\label{QK_con}
\frac{1}{2}\beta^{\T}\bK\beta \leq t.
\end{equation}

The constraints (\ref{hl_con}) and (\ref{QK_con}) are dependent on the particular instantiation of the stochastic kernel matrix $\bK_{N}$ and we thus advocate the strategy of satisfying them with high probability.
We denote the probability of violation of {\em any} one constraint $\IE_{i}$ in (\ref{hl_con}) by $\delta_{1}\ll1$ and the probability of violating the constraint (\ref{QK_con})
by $\delta_{2}\ll1$. The intersection of all the inequalities $\IE_{i}$ is satisfied with a high probability $1-\delta_{1}$, which gives us a sufficient condition
that each inequality be satisfied with the higher probability of $1 - \delta_{1}/m$ (union bound \ref{sec:union_bound})
\begin{equation}
\label{cc-1}
\Prob\bigg(\xi_{i}\geq 1 - y_{i}\big(\sum_{j}\bK_{ij}\beta_{j} + b\big)\bigg) \geq 1 - \frac{\delta_{1}}{m},
\end{equation}
and
\begin{equation}
\label{cc-2}
\Prob\bigg(\frac{1}{2}\beta^{\T}\bK\beta \leq t\bigg) \geq 1 - \delta_{2}.
\end{equation}
The sufficient conditions for the constraints (\ref{cc-1}) and (\ref{cc-2}) lead to Theorem \ref{thm:rob}, and is
covered in the Proofs Section.
At optimality, we shall denote $J(t,\xi,\beta,b;\K) = J^{*}(\K)$, dropping the optimal variables $t^{*},\xi^{*},\beta^{*},b^{*}$ from the
notation. 

\subsection{Second-order cone programming}
\label{sec:socp}

A second-order cone program (SOCP) involves a linear objective function with one or more second-order cone constraints
and any additional linear inequality constraints. An example of a second-order cone constraint in $d+1$ dimensions is $\norm{x}_{2}\leq t$, where $\bx\in\RR^{d},\ t\in\RR$.
The convex approximations we developed for the chance constraints
(\ref{cc-1}) and (\ref{cc-2}), i.e.\ (\ref{suff_cc-1}) and (\ref{suff_cc-2}) respectively, are second-order cone representable. This was shown for the chance
constraint approach involving the SVM Dual in Ref.\ \cite{Bhadra2010}. We show it for the SVM Primal \ref{eq:rob_primal} below.

The constraint (\ref{suff_cc-1}) can be represented as a second-order cone (SOC) constraint by simply introducing an
additional variable $r$ with the SOC constraint $\norm{\beta}\leq r$, and replacing $\norm{\beta}$ by $r$ in (\ref{suff_cc-1}) which is an additional linear inequality constraint.
Similarly, (\ref{suff_cc-2}) is a quadratic constraint which can be straight-forwardly expressed as a SOC constraint $\norm{A\beta}\leq t$, where $A\in\psdset$ is
the positive-definite square root of the matrix $A^{2}=\left(\K^{*} + \frac{\c}{2\sqrt{N}}\kappa(\delta_{2}) I  \right)\in\psdset$. The primal objective (\ref{eq:obj2})
along with the linear and SOC constraints
stated above forms a second-order cone program (SOCP) which can be solved efficiently \cite{Lobo1998socp}. We use the MOSEK solver for its numerical implementation in Python
using the convex programming package \texttt{cvxpy} \cite{diamond2016cvxpy,agrawal2018rewriting}.

\section{Proofs}
\label{sec:proofs}

\begin{proof}[Proof of Theorem \ref{thm:N_bound}]
Let
\begin{equation}
\begin{aligned}
g^{*}(\bx) &= \sum_{j}\beta^{*}_{j} K^{*}(\bx,\bx_{j}) + b^*\\
g^{(N)}(\bx)&= \sum_{j}\beta^{*}_{j} K^{(N)}(\bx,\bx_{j}) + b^*,
\end{aligned}
\end{equation}
so that the margin of classification $\gamma(\bx) = y g^{*}(\bx)$ for EKC at point $(\bx,y)$.

The reliability $\R(\bx;f^{(N)},f^{*})\geq1-\delta$ means that $g^{(N)}(\bx)$ and $g^{*}(\bx)$
have the same sign with probability of at least $1-\delta$. Equivalently, the probability that they have different signs
\begin{equation}
\Prob\left(g^{(N)}(\bx)g^{*}(\bx)<0\right)\leq\delta
\end{equation}
does not exceed $\delta$. Use $g^{(N)}(\bx) = g^{*}(\bx) - \sum_{j}\beta_{j}\Delta K(\bx,\bx_{j})$.
Rearranging,
$$\Prob\left(g^*(\bx)\sum_{j}\beta_{j}\Delta K(\bx,\bx_{j})> g^{*}(\bx)^{2} \right)\leq\delta.$$
Direct application of Lemma \ref{lemma:tail} using the uncertainty model \ref{eq:um}, and Definition \ref{lemma:cf} for the variance of $\Delta K(\bx,\bx_{j}) $, completes the proof.
\end{proof}

Theorem \ref{thm:emp_risk} follows a procedure similar to that of Theorem \ref{thm:N_bound} but additionally uses the Union bound
mentioned in Section \ref{sec:union_bound}.
\begin{proof}[Proof of Theorem \ref{thm:emp_risk}]
We demand $S_{0}[f^{(N)}]\subseteq S_{\gamma}[f^*]$ over the dataset $\D$, with high probability:
\begin{equation}
\label{eq:subset_orig}
\Prob\left(S_{0}[f^{(N)}]\subseteq S_{\gamma}[f^*]\right)\geq 1-\delta.
\end{equation}
$(\bx,y)\in S_{\gamma}[f^*]$ if $yg^{*}(\bx)-\gamma<0$ by Definition \ref{margin_violation}.
Similarly $(\bx,y)\in S_{0}[f^{(N)}]$ if $yg^{(N)}(\bx)<0$.

Equivalently, we bound the probability that $(\bx,y)\in S_{0}[f^{(N)}]$ given $(\bx,y)\in \overline{S_{\gamma}[f^*]}$.

$(\bx_{i},y_{i})\in S_{0}[f^{(N)}]$ gives
\begin{equation}
\begin{aligned}
&y_{i}g^{(N)}(\bx_{i})<0\\
&y_{i}g^{*}(\bx_{i}) - y_{i}\left(\sum_{j}\Delta K^{(N)}_{ij}\beta_{j}\right)<0,\\
\end{aligned}
\end{equation}
where $\Delta K^{(N)}_{ij} = K^{*}(\bx_{i},\bx_{j})-K^{(N)}(\bx_{i},\bx_{j})$.
Now, using $y_{i}g^{*}(\bx_{i})-\gamma>0$, i.e. $(\bx_{i},y_{i})\in \overline{S_{\gamma}[f^*]}$, the above equation gives
\begin{equation}\tag{$\bar{A}_{i}$}
\label{eq:A_i}
y_{i}\left(\sum_{j}\Delta K^{(N)}(\bx_{i},\bx_{j})\beta_{j}\right)>\gamma,
\end{equation}
whose probability we would like to make smaller than $\delta$ (for all datapoints $\bx_{i}\in\D$).

Using the union bound in \ref{sec:union_bound}, for all datapoints in $\bx_{i}\in\D$, we shall
make the combined probability of \ref{eq:A_i} small
\begin{equation}
\Prob(\cup_{i} \bar{A}_{i})\leq \sum_{i=1}^{M}\Prob(\bar{A}_{i})\leq \delta,
\end{equation}
which is satisfied if
\begin{equation}
\label{eq:A_i_2}
\Prob\left(y_{i}\left(\sum_{j}\Delta K^{(N)}_{ij}\beta_{j}\right)>\gamma\right)\leq\delta/M.
\end{equation}
Applying Lemma \ref{lemma:tail} to
and using the uncertainty model \ref{eq:um}, and Definition \ref{lemma:cf} for the variance of $\Delta K(\bx_{i},\bx_{j}) $, completes the proof.


%
\end{proof}

\begin{proof}[Proof of Corollary \ref{cor:Reliability}]
We have from Eq.\ (\ref{eq:max_gamma})
\begin{equation}
\epsilon_{\gamma*}\left(f^{*}\right) = \epsilon_{0}\left(f^*\right).
\end{equation}
Taking $\delta/M $ as $ \delta$ in Eq.\ (\ref{eq:A_i_2}) implies that
\begin{equation}
\R_{\delta}\left(f^{(N)},f^{*}\right) = 1.
\end{equation}
That is, each datapoint in $\D$ is reliably classified by \ref{eq:skc*}
with a probability of at least $1-\delta$.
\end{proof}
%

\begin{proof}[Proof of Theorem \ref{thm:rob}]
\label{proof:thm:rob}

The optimization problem (\ref{eq:obj2}) subject to the chance constraints (\ref{cc-1}) and (\ref{cc-2}) is in general nonconvex and we now derive
sufficient conditions for the chance constraints (\ref{cc-1}) and (\ref{cc-2}) using the uncertainty model \ref{eq:um}.
We first write the random kernel $K$ as the sum of its expected and stochastic parts $ K_{ij} = K^{*}_{ij} - \Delta K^{(N)}_{ij} $
in the chance constraint (\ref{cc-1})
\begin{equation}
\label{eq:cc_1_supp}
\Prob\bigg( y_{i}\sum_{j}\beta_{j}\Delta K^{(N)}_{ij}\leq h\bigg)\geq 1-\frac{\delta_{1}}{m},
\end{equation}
where $h =  y_{i}\big(\sum_{j} K^{*}_{ij}\beta_{j} + b\big) + \xi_{i} - 1$.
We derive a tail bound for the constraint violation probability for the random variable $X = y_{i}\sum_{j}\beta_{j}\Delta K^{(N)}_{ij}$
using lemma \ref{lemma:tail} and find the sufficient condition
\begin{equation}
\label{suff_cc-1}
\xi_{i}\geq 1 - y_{i}\bigg(\sum_{j} K^{*}_{ij}\beta_{j} + b\bigg) + \frac{\c}{2\sqrt{N}} \kappa\bigg(\frac{\delta_{1}}{m} \bigg) \norm{\beta},
\end{equation}
for satisfying the chance constraint (\ref{cc-1}).

We similarly derive a sufficient condition for the chance constraint (\ref{cc-2}). With the kernel matrix $\K_{ij}=\K^{*}_{ij} - \Delta \K^{(N)}_{ij}$,
we define the random variable
\begin{equation}
R = -\frac{1}{2}\beta^{\T}\Delta{\K}^{(N)}\beta =  -\frac{1}{2}\sum_{ij}\Delta{K}^{(N)}_{ij}\beta_{i}\beta_{j}.
\end{equation}
We rewrite (\ref{cc-2}) as $\Prob(R\leq t')\leq \delta_{2}$
and apply Lemma \ref{lemma:tail} to find the sufficient condition. We detail it below for clarity:
%
\begin{equation}
\label{suff_cc-2}
\begin{aligned}
t'&\geq \frac{\c}{4\sqrt{N}}\norm{\beta}^{2} \kappa(\delta_{2})\\
t&\geq \frac{1}{2}\beta^{\T}\K^{*}\beta + \frac{\c}{4\sqrt{N}}\norm{\beta}^{2} \kappa(\delta_{2})\\
&=\frac{1}{2}\beta^{\T}\bigg(\K^{*} + \frac{\c}{2\sqrt{N}}\kappa(\delta_{2}) I  \bigg)\beta
\end{aligned}
\end{equation}
for satisfying the chance constraint (\ref{cc-2}).

Use sufficient conditions (\ref{suff_cc-1}), (\ref{suff_cc-2}) for the objective
(\ref{eq:obj2}) and eliminate (dummy) variable $t$. Apply the union bound (\ref{sec:union_bound}) for conditions (\ref{cc-1}) and (\ref{cc-2})
to show $\forall\beta,b$:
\begin{equation}
\Prob\big(J(\beta,b;\K)\leq J_{\rob}(\xi,\beta,b;\K^{*})\big)\geq 1 - \delta_{1} - \delta_{2}.
\end{equation}
At optimality $J^{*}(\K) = \min_{\beta,b}J(\beta,b;\K)$ and $J^{*}_{\rob} = \min_{\xi,\beta,b} J_{\rob}(\xi,\beta,b;\K^{*})$ to get (\ref{eq:whp_rob}).
Eliminating $\xi\geq 0$ from the constraints (\ref{suff_cc-1}) give us the hinge-loss form of the objective used in Eq.\ (\ref{eq:J_rob}).
\end{proof}

\begin{proof}[Proof of Theorem \ref{thm:rob_est}]
The \ref{eq:rob_primal} Program has the conditions (\ref{suff_cc-1}) and (\ref{suff_cc-2}) which depend on the exact kernel matrix
$\K^{*}$. In the setting of \ref{eq:J_rob_l2_est}, we do not have access to the exact kernel matrix and instead work with an estimated kernel matrix $\widehat\K$. We write \ref{suff_cc-1} as
\begin{equation}
\begin{aligned}
\label{ssc_1}
\xi_{i}\geq 1 -& y_{i}\bigg(\sum_{j} \widehat{\K}_{ij}\beta_{j} + b\bigg) + \frac{\c}{2\sqrt{N}} \kappa\bigg(\frac{\delta_{1}}{m} \bigg) \norm{\beta}\\
-& y_{i}\sum_{j}\bigg({\K}^{*}_{ij}-\widehat{\K}_{ij} \bigg) \beta_{j}.
\end{aligned}
\end{equation}
Now, treating the ${\K}^{*}_{ij}-\widehat{\K}_{ij} = \Delta K^{(T)}_{ij}$ as a random variable allows us
to satisfy \ref{suff_cc-1} with high probability $1-\delta'_{1}/m$, whenever
\begin{equation}
\label{eq:ssc1_supp}
\begin{aligned}
\xi_{i}\geq 1 -  y_i\left(\sum_{j=1}^m \beta_{j} \widehat{\K}_{ij}+b\right) &+ \frac{\c}{2\sqrt{N}} \kappa\left(\frac{\delta_{1}}{m}\right) \norm{\beta}_{2}&\\
&+ \Delta\left(\frac{\delta'_{1}}{m},T\right)\norm{\beta}_{2}.&
\end{aligned}
\end{equation}
The above equation is sufficient to ensure that (\ref{suff_cc-1}) holds with high probability when
we are using the estimate in place of the exact kernel matrix. It also employed Lemma \ref{lemma:tail}
and the union bound identically as done previously from (\ref{eq:cc_1_supp}).

Similarly, (\ref{suff_cc-2}) is written as
\begin{equation}
t\geq\frac{1}{2}\beta^{\T}\bigg(\widehat{\K} + [\K^{*}-\widehat{\K}] + \frac{\c}{2\sqrt{N}}\kappa(\delta_{2}) I  \bigg)\beta.
\end{equation}
Now we know that the term in the square bracket can be bounded with high probability. Thus, (\ref{suff_cc-2})
is satisfied with a probability of $1-\delta'_{2}$, whenever
\begin{equation}
\label{eq:ssc2_supp}
t\geq\frac{1}{2} \beta^{\top}\left(\widehat{\bK} + \frac{\c}{2\sqrt{N}}\kappa(\delta_{2})  \mathbf{I}_{m\times m}+
\Delta(\delta'_{2},T)  \mathbf{I}_{m\times m}\right)\beta.
\end{equation}
(\ref{eq:ssc1_supp}) and (\ref{eq:ssc2_supp}) together with the union bound shows that
\begin{equation}
\Prob\left(J*_{\rob}\leq \hat{J}^{*}_{\rob} \right)\geq 1-\delta'_{1}-\delta'_{2}.
\end{equation}
\end{proof}
\section{Bounds from precise kernel evaluation}
\label{sec:precise}

\begin{thm}[Precise Kernel Estimation]
Let $m_{\sv}$ denote the number of support vectors of an SVM.
Over any dataset $\D$ of size $|\D|=M$, let $\K^{(N)}$ and $\K^*$ denote
the $m_{\sv}\times M $ size matrices $\K^{(N)}_{ij} = K^{(N)}_{ij}$, $\K^{*}_{ij} = K^{*}_{ij}$,
$i\in[m_{\sv}],\ j\in[M]$, whose entries are the stochastic kernel and the exact kernel functions respectively.
Then the number of measurements
\begin{equation}
\label{eq:precise}
N\geq \frac{\c^{2}}{2\varepsilon^{2}} m_{\sv}M\log\left(\frac{m_{\sv}M}{\delta}\right)
\end{equation}
is sufficient to ensure
\begin{equation}
\norm{\K^{(N)}-\K^{*}}_{F}\leq\varepsilon 
\end{equation}
holds with a probability of at least $1-\delta$.
\end{thm}
\begin{proof}
We would like for the Frobenius norm
\begin{equation}
\sqrt{\sum_{ij} \left(K^{(N)}_{ij}-K^{*}_{ij}\right)^{2}}\leq\varepsilon
\end{equation}
to be bounded with high probability. For the above condition, it is sufficient to
ensure that each term is bounded by
\begin{equation}
\IE_{ij}:\ |{\K^{(N)}_{ij}-\K^{*}_{ij}}|\leq \frac{\varepsilon}{\sqrt{m_{\sv}M}}.
\end{equation}
We would like {\em all} inequalities $\IE_{ij}$ to hold with high probability.
In other words, we shall ensure that its compliment $ \overline{\cap_{ij}\IE_{ij}} = \cup_{ij}\overline{\IE_{ij}} $
is bounded by a low probability
\begin{equation}
\begin{aligned}
\Prob\left(\overline{\cap_{ij}\IE_{ij}}\right) &= \Prob\left(\cup_{ij}\overline{\IE_{ij}}\right)\\
&\leq m_{\sv}M \Prob\left(\overline{\IE_{ij}}\right)\\
&\leq \delta
\end{aligned}
\end{equation}
Here, the second and third line together with Lemma \ref{lemma:tail} prove the result.
\end{proof}

Using the above Theorem, we can get the bound stated in Havlicek et.\ al.\ \cite{Havlicek2019},
by taking $m_{\sv}=M=m$. One can now see the main advantage of Theorem \ref{thm:emp_risk}: the margin $\gamma$ is
not vanishingly small like the precision $\varepsilon$. Additionally, we improve the dependence on $M$ exponentially.

\subsection{Bound from the empirical risk}
\addtocounter{thm}{-3}
\addtocounter{cor}{2}
\begin{cor}
\label{cor:emp_risk}
Over a training set of size $m$, the empirical risk of \ref{eq:skc*}
\begin{equation}
\epsilon_{0}\left(f^{(N)}\right)\leq\epsilon_{1}\left(f^*\right)
\end{equation}
is bounded by the {\em margin error} of \ref{eq:ekc*}, with probability $1-\delta$, whenever
\begin{equation}
\label{N_margin}
N\geq \frac{\c^{2}}{2}\norm{\beta^{*}}^{2} \log \frac{m}{\delta}.
\end{equation}
\end{cor}
The margin errors of \ref{eq:ekc*} are listed in the right-most column of Table \ref{tab:N_prac}.
The corresponding mark is a black (dash-dotted) vertical line in Figure \ref{fig:N_prac}.
We stress that these values of $N$ are quite small compared to what is needed to ensure reliable reproduction of \ref{eq:ekc*} (green region).
The corollary above ensures a small empirical risk for \ref{eq:skc*}.

Often $m_{\sv}\ll m $ can be taken as a constant not dependent on $m$.
When applied over the training set itself, the number of circuit evaluations per kernel entry scales as $N\sim \log m/\gamma^{2}$.
Demanding $\norm{\K^{(N)}-\K^{*}}\leq \epsilon$ in the operator or Frobenius distance gives $N \sim m^{2}/\varepsilon^{2}$ scaling \cite{Havlicek2019}. Note again, the parameter $\gamma$ in Theorem \ref{thm:emp_risk}
plays a role analogous to the precision $\varepsilon$, but is not vanishingly small.

\subsection{Benchmarks in literature}

As we mentioned in the main article, clearly stated benchmarks for $N$ are not easily found.
There exists also a large discrepancy in the number of measurements used in various studies of quantum classifiers \cite{Wang2021nisq,Hubregtsen2021,Havlicek2019}.
In the freely accessible IBM quantum processors, the default number of measurements per job is set to $1024$ (since 2023 it is $4000$, see \url{https://docs.quantum-computing.ibm.com/admin/faq-admin}).
While Refs.\ \cite{Wang2021nisq,Hubregtsen2021} work close to these restrictions, Ref.\ {\cite{Havlicek2019}} works with a much larger 50,000 measurements per kernel entry.
The error mitigation technique used in \cite{Havlicek2019} results in higher requirements.
Ref.\ \cite{Hubregtsen2021} used 175 measurements to obtain the training kernel matrix and get good accuracies over the test set, but it is not clear the number of measurements used for prediction.
Ref.\ \cite{Wang2021nisq} uses the same number of measurements for both the training and the testing phases. The range of observed accuracies in the latter is wide, indicating that the classification is unreliable.

In our simulations using the \ref{eq:J_rob_l2_est} Program, we use a confidence interval of $\Delta(T)=0.1$. ($ |K^{*}_{ij}-K^{(T)}_{ij}|\leq0.1$, with at least $99\%$ probability, for each matrix entry). This results in a modest shots requirement
of only about $T\approx400$. 
Our work indicates that the number of measurement shots is far more crucial during the classification phase (for prediction) than during the training phase.

\section{Device Noise}
\label{sec:depol}
\begin{table*}
\def\arraystretch{1.5}
\begin{tabular}{|c|c|c|c|c|c|}
\hline
Number &\thead{Training\\ Kernel\\ Matrix }& \thead{Optimization problem}
  & \thead{Noisy Classifier} & \thead{Ideal classifier} & \thead{Accuracy of Ideal \\
 Classifier over $\D_{\train}$}\\
 \hline
1&  $\hat{\K}^{+}$ & $\beta^{*},b^*$ from \ref{opt:p}
 & \thead{U-SKC\\ $f^{(N,\lambda)}_{u}(\bx) = f(\bx|\beta^{*},b^{*},K^{(N,\lambda)})$} & \thead{U-EKC\\ $f^{*}_{u}(\bx) = f(\bx|\beta^{*},b^{*},K^{*})$} & 80\% \\ \hline
2& $\hat{\K}_{\textrm{miti}}^{+}$ & $\beta^{*},b^*$ from \ref{opt:p}
 & \thead{M-SKC\\ $f^{(N,\lambda)}_{m}(\bx) = f(\bx|\beta^{*},b^{*},K^{(N,\lambda)}_{\textrm{miti}})$} &
 \thead{M-EKC\\ $f^{*}_{m}(\bx) = f(\bx|\beta^{*},b^{*},K^{*})$} & 97.5\%\\
 \hline
3& $\hat{\K}$ & $\beta^{*},b^*$ from \ref{eq:J_rob_l2_est}& \thead{U-RSKC\\ $h^{(N,\lambda)}_{u}(\bx) = f(\bx|\beta^{*},b^{*},K^{(N,\lambda)})$} & \thead{U-REKC\\$h^{*}_{u}(\bx) = f(\bx|\beta^{*},b^{*},K^{*})$} & 100\% \\
 \hline
4&$\hat{\K}_{\textrm{miti}}$ & $\beta^{*},b^*$ from \ref{eq:J_rob_l2_est}&
\thead{M-RSKC\\ $h^{(N,\lambda)}_{m}(\bx) = f(\bx|\beta^{*},b^{*},K^{(N,\lambda)}_{\textrm{miti}})$} &
\thead{M-REKC\\ $h^{*}_{m}(\bx) = f(\bx|\beta^{*},b^{*},K^{*})$} & 100\% \\
\hline
5&$\K^{*}$ & $\beta^{*},b^*$ from \ref{opt:p}&
\thead{---} &
\thead{EKC\\ $f^{*}(\bx) = f(\bx|\beta^{*},b^{*},K^{*})$} & 100\% \\
\hline
\end{tabular}
\captionsetup{justification=raggedright}
\caption{ Our experiment here studies how reliably each of these noisy classifiers reproduces its ideal counterpart. The ideal classifiers have access to the exact kernel function during the classification but have been trained using a (noisy) estimated training kernel matrix.
The training kernel matrix $\hat{K}_{ij} = K^{(T,\lambda)}_{ij} $ is measured using $T=400$ shots and is fixed. The depolarizing parameter was set to $\lambda=0.05$. It is error-mitigated for the classifiers (2) and (4) using the M-MEAN method in Ref.\ \cite{Hubregtsen2021}. Classifier (5) represents the ideal classifier which has been trained using the exact training kernel matrix $\K^{*}$ and has access to the exact kernel function $K^{*}$ during the classification.}
\label{tab:rob_zoo}
\end{table*}

\begin{figure*}
    \centering
\centering
    \begin{subfigure}[hb]{0.48\textwidth}
    \centering
    \includegraphics[width=\textwidth,trim = 0cm 0cm 0cm 0cm]{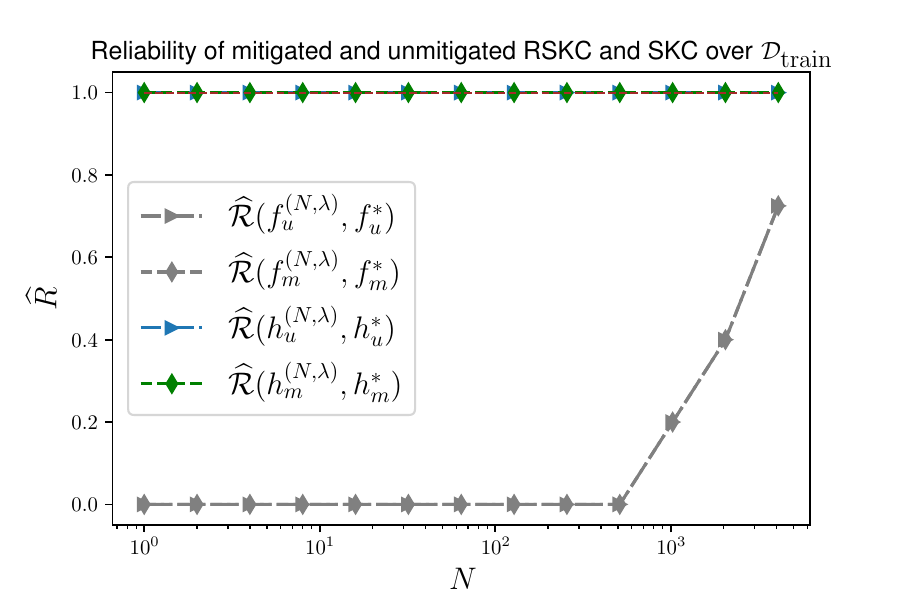}
    \caption{}
    \label{fig:rel_compare_depol}
    \end{subfigure}
\hspace{-0.08cm}
    \begin{subfigure}[hb]{0.48\textwidth}
    \centering
    \includegraphics[width=\textwidth,trim = 0cm 0cm 0cm 0cm]{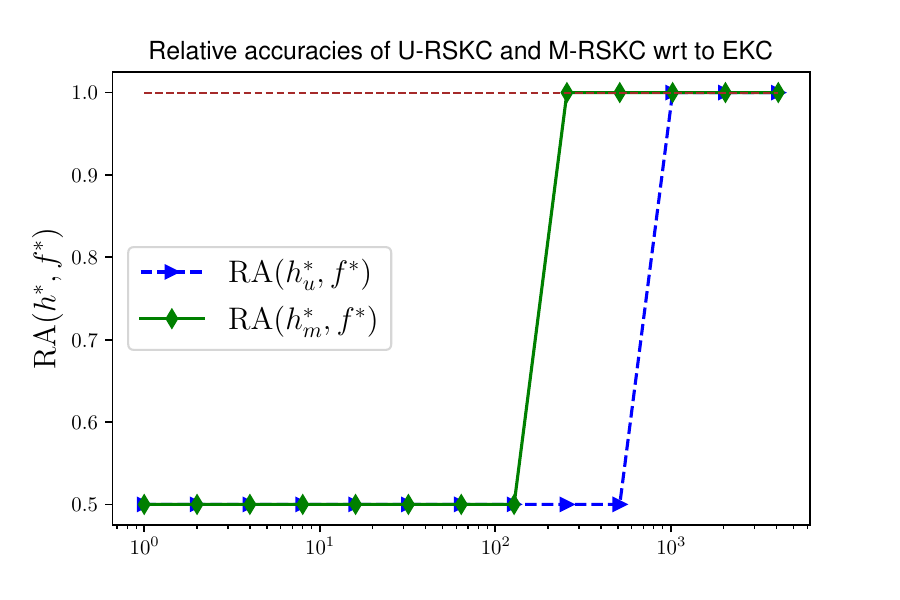}
    \caption{}
    \label{fig:rel_acc_depol}
    \end{subfigure}
\captionsetup{justification=raggedright}
\caption{(a) The robust classifiers $h^{(N,\lambda)}_{u}$ and $h^{(N,\lambda)}_{m}$ are far more reliable than the nominal SVM classifiers $f^{(N,\lambda)}_{u}$ and $f^{(N,\lambda)}_{m}$.
They reliably reproduce the output of their ideal counterparts $h^{*}_{u}$ and $h^{*}_{m}$ which have the same parameter values (i.e., $\beta^{*},b^*$) for all values of $N$. The number of independent trials $N_{\trials}=100$.
(b) Shows the relative accuracies of unmitigated ($h^{*}_{u}$) and mitigated ($h^{*}_{m}$) robust exact classifiers (REKC) as a function of $N$. The error mitigation helps the REKC reproduce the results of the ideal classifier $f^{*}$ for smaller $N$, but has no effect on the reliabilities of either the robust or nominal stochastic classifiers.}
\label{fig:rel_depol}
\end{figure*}

In this Section, we include a device noise using the depolarizing noise model and the error mitigation scheme employed in Ref.\ \cite{Hubregtsen2021}. 
The depolarizing noise is modeled using a quantum channel $D_{\lambda}: \mathcal{H}\otimes\mathcal{H^*}\to\mathcal{H}\otimes\mathcal{H^*} $ which maps a density matrix $\rho\in\mathcal{H}\otimes\mathcal{H^*}$ to
\begin{equation}
D_{\lambda}[\rho] = (1-\lambda) \rho + \frac{\lambda}{d}\mathbb{1},
\end{equation}
where $d^{2} = \textrm{dim}(\mathcal{H}\otimes\mathcal{H^*})$ is the dimension of the feature Hilbert space.
This leads to device kernel evaluations 
\begin{equation}
K^{(\textrm{dev})}_{ij} = (1-\lambda_{i}\lambda_{j})K_{ij} + \lambda_{i}\lambda_{j} \frac{1}{d},
\end{equation}
from which the noiseless kernel $K_{ij}$ will have to be inferred using one of the strategies put forth in Ref.\ \cite{Hubregtsen2021}. In our simulations,
we use the M-MEAN strategy.

In our simulations, we set $\lambda_{i}=\lambda=0.05$. The noisy kernel which is evaluated using $N$ measurement shots and a depolarizing parameter $\lambda$ is
denoted as $K^{(N,\lambda)}_{ij}= K^{(N,\lambda)}(\bx_{i},\bx_{j}) $. The training kernel matrix estimated using $T$ measurement shots and a depolarizing parameter $\lambda$ is denoted as
\begin{equation}
\widehat{\K}_{ij} = K^{(T,\lambda)}_{ij},\ \ i,j\in[m].
\label{eq:est_depol}
\end{equation}
After mitigating the error, we denote the training kernel matrix as $\hat{\K}_{\textrm{miti}}$.

Note that the matrix in (\ref{eq:est_depol}) will not be positive semidefinite. However, the nominal SVM optimization problem \ref{opt:p} requires a positive semidefinite matrix for it to be well posed. We shall denote its positive semidefinite part by $\widehat{\K}^{+}$ which
subtracts the largest negative eigenvalue $\kappa_{\min}$ to the spectrum. Here we refer to it as a spectral shift \cite{Wang2021nisq,Hubregtsen2021}:
\begin{equation}
\widehat{\K}^{+}= shift({\widehat{\K}}) := \begin{cases}
    \widehat{\K} - \kappa_{\min}\mathbb{1} & \text{if } \kappa_{\min}< 0\\
    \widehat{\K}             & \text{otherwise}.\end{cases}
\end{equation}
Similarly, the mitigated kernel matrix need not be positive semidefinite and we denote its positive semidefinite
part as
\begin{equation}
\widehat{\K}_{\textrm{miti}}^{+}= shift({\widehat{\K}_{\textrm{miti}}}).
\end{equation}

Note that the Robust Program (\ref{eq:J_rob_l2_est}) which we derived
using an estimated kernel matrix $\widehat{\K}$ does not need this ``fixing of the spectrum."
We now train the following classifiers:
\begin{enumerate}
\item Using the training kernel matrix $shift({\widehat{\K}})$, solve \ref{opt:p} to find the optimal parameters $\beta^{*},b^{*}$.
We then compare the resulting stochastic and exact kernel classifiers U-SKC and U-EKC, where U stands for ``unmitigated",
 respectively denoted by $f^{(N,\lambda)}_{u}$ and $f^{*}_{u}$.
\item Using the training kernel matrix $shift({\widehat{\K}}_{\textrm{miti}})$, solve \ref{opt:p} to find the optimal parameters $\beta^{*},b^{*}$. We compare the resulting stochastic and exact kernel classifiers as M-SKC and M-EKC, where M stands for ``mitigated", respectively denoted by $f^{(N,\lambda)}_{m}$ and $f^{*}_{m}$.
\item Using the training kernel matrix $\widehat{\K}$, and with specified values of parameters $N$ and $\delta_{1}= \delta_{2}= \delta_{2}'=0.01$, and $\Delta(\delta'_{1}/m,T)=0.1$,
solve (\ref{eq:J_rob_l2_est}) to find the optimal parameters $\beta^{*},b^{*}$.
We now compare the resulting robust stochastic and robust exact kernel classifiers as U-RSKC and U-REKC, respectively
denoted by $h^{(N,\lambda)}_{u}$ and $h^{*}_u$.
\item Using the training kernel matrix $\widehat{\K}_{\textrm{miti}}$, and with specified values of parameters $N$ and $\delta_{1}= \delta_{2}= \delta_{2}'=0.01$,
and $\Delta(\delta'_{1}/m,T)=0.1$, 
solve (\ref{eq:J_rob_l2_est}) to find the optimal parameters $\beta^{*},b^{*}$.
We now compare the resulting robust stochastic and robust exact kernel classifiers as M-RSKC and M-REKC, respectively
denoted by $h^{(N,\lambda)}_{m}$ and $h^{*}_{m}$.
\end{enumerate}
We now wish to learn the performance of the above classifiers over the training dataset $\D_{\train}$ itself.
The reliabilities of the above four stochastic classifiers are shown in \figurename\ \ref{fig:rel_compare_depol}. Even in the presence of the depolarizing noise, the robust classifiers far outperform the nominal classifiers on both the accuracy of EKC and the reliability of SKC. Note that the error mitigation has no effect on the reliability (see Fig.\ \ref{fig:rel_compare_depol}).

The error mitigation, however, helps the robust classifier $h^{*}_{m}$ reproduce the labels of EKC $f^{*}$ for a smaller value of $N$, as seen from \ref{fig:rel_acc_depol}.
Using the error mitigation and the robust version of the classifier (M-RSKC) reliably reproduces EKC for $N=2^{8}$, which is at par with the robust classifier when there is no depolarizing noise (RSKC, see Table \ref{tab:Reliability}). Without the error mitigation, the robust classifier needs $N=2^{10}$ measurements to reproduce EKC.

In this Section, the Robust Program \ref{eq:J_rob_l2_est} was used without any modification and the same confidence interval of $\Delta=0.1$ ensured its robustness even in the presence of the depolarizing noise.

\section{More on Accuracy and Reliability}
\label{sec:more_rel}
In this Section, we empirically justify the idea that the noise is far more relevant to the classification phase than for the training phase. To this end, we look at the training kernel matrix which is estimated using $T$ measurement shots per entry.
Following the notation in the preceding section, we train the SVM classifier using the matrix $\widehat{\K}^{+}$, where $\widehat{\K}_{ij} = K^{(T)}_{ij}$ is given by (\ref{eq:est_depol}) with no depolarizing noise, i.e., $\lambda=0$. We only take one instantiation of the training kernel matrix and fix it to $\widehat{\K} $, as done in the previous section and in Section \ref{sec:stat_est}.
We then represent the EKC and SKC using the notation $f^{*}_{T}$ and $f^{(N)}_{T}$ respectively, where the subscript $T$ denotes the shots used for estimating the training kernel matrix. By varying $T$, we compare their performance over an independent test set.
\begin{figure}
    \centering
\centering
    \begin{subfigure}[hb]{0.48\textwidth}
    \centering
    \includegraphics[width=\textwidth,trim = 0cm 0cm 0cm 0cm]{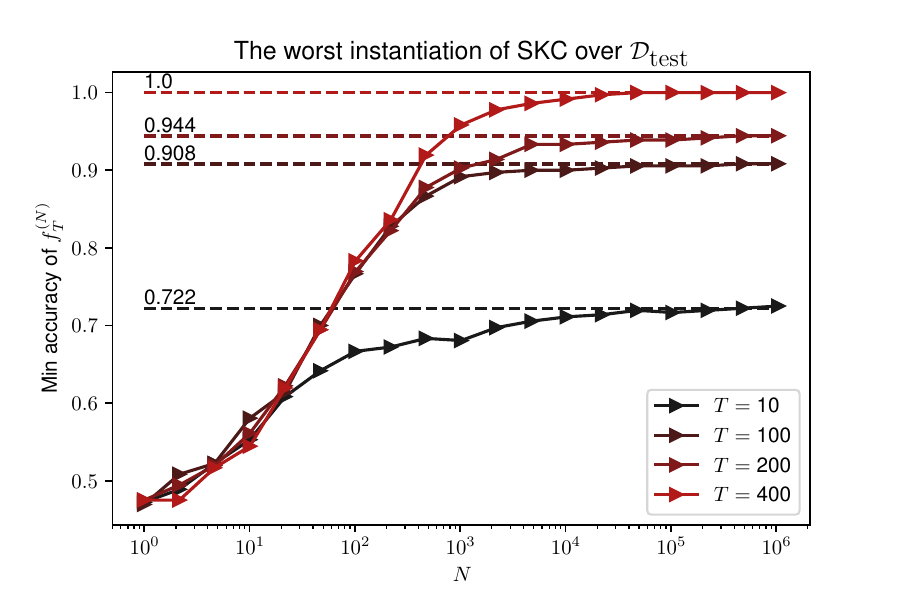}
    \caption{}
    \label{fig:worst_acc}
    \end{subfigure}
\hspace{-0.08cm}
    \begin{subfigure}[hb]{0.48\textwidth}
    \centering
    \includegraphics[width=\textwidth,trim = 0cm 0cm 0cm 0cm]{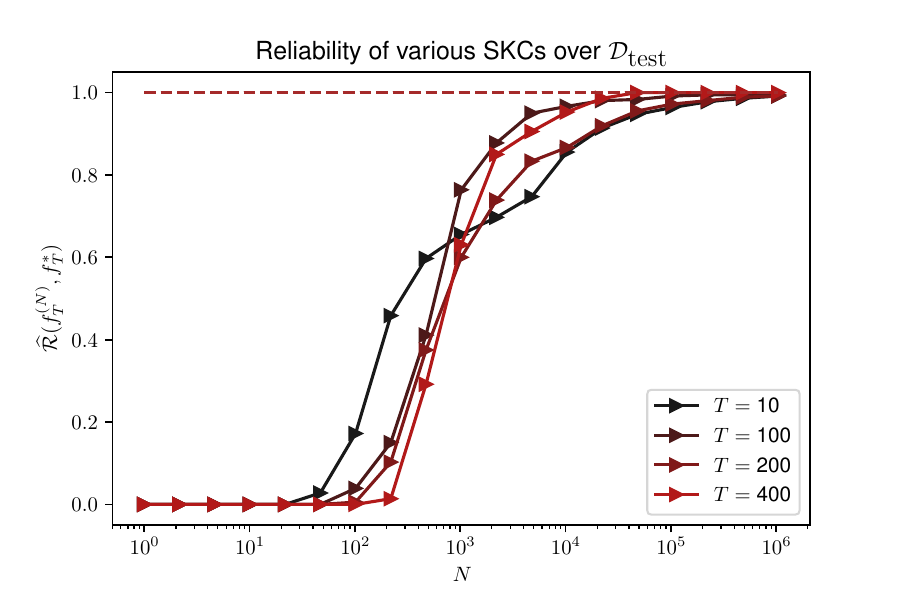}
    \caption{}
    \label{fig:T_rel}
    \end{subfigure}
\captionsetup{justification=raggedright}
\caption{(a) The accuracy of worst performing SKC instantiation, among $N_{\trials}=200$ independent trials, is shown as a function of $N$. The number of shots during training $T$ is varied. The horizontal dashed lines represent the accuracy of the corresponding EKC $f^{*}_{T}$. For large $N$, SKC has the same accuracy as the corresponding EKC.
(b) Shows the reliabilities of the SKCs $f^{(N)}_{T}$ with respect to their respective EKCs $f^{*}_{T}$. The noise during the training phase (as parametrized by $T$) does not affect the reliability of classification.}
\label{fig:more_rel}
\end{figure}

The accuracy of EKC for $T=10, 100, 200, 400$ are respectively 72.2\%, 90.8\%, 94.4\% and 100\% as indicated by the horizontal dashed lines in Fig.\ \ref{fig:worst_acc}. Here the minimum accuracy instantiation of SKC, over $N_{\trials}=200$ independent instantiations, is shown as a function of $N$.

In some sense, the reliability tells us about the worst-performing SKC instantiation: if the latter produces the same labels as EKC, then the SKC can be considered reliable. Noise during training is clearly seen to affect the accuracy of EKC but does not have an effect on the reliability of SKC (Fig.\ \ref{fig:T_rel}). In particular, EKC with $T=400$ has 100\% accuracy over the test dataset. However, the number of measurements $N$ needed for the nominal SKC to reliably reproduce EKC is $N=2^{14}\gg400$ (see Table  \ref{tab:test_reliability}).

\end{appendix}

\end{document}